\documentclass[12pt]{article}
\usepackage{amsmath,amssymb}
\usepackage{graphicx,psfrag,epsf}
\usepackage{enumerate}
\usepackage{natbib}
\usepackage{color,xcolor}
\graphicspath{ {Images/} }
\usepackage{booktabs}
\usepackage[ruled,vlined]{algorithm2e}
\usepackage[english]{babel}
\usepackage{subcaption}
\usepackage{amsthm}

\newcommand{\dd}{\mathrm{d}}
\newcommand{\ta}{\theta}
\newcommand{\EE}{{\mathbb E}}
\newcommand{\R}{\mathbb R}
\newcommand{\ld}{\lambda}
\newcommand{\pl}{D}

\newcommand{\blind}{0}

\def\cM{ {\mathcal M} }

\def\rg{{\rm g}}

\newcommand{\cov}{\mbox{Cov}}

\newcommand{\T}{{\mbox{\scriptsize \sf T}}}

\newtheorem{definition}{Definition}

\newtheorem{corollary}{Corollary}
\newtheorem{proposition}{Proposition}
\newtheorem{remark}{Remark}

\begin{document}

\def\spacingset#1{\renewcommand{\baselinestretch}%
{#1}\small\normalsize} \spacingset{1}

\if0\blind
{
  \title{\bf $K$-functions for point processes on complex surfaces}
  \author{Francisco Cuevas-Pacheco\hspace{.2cm}\\
    Departamento de Matemática, Universidad T{\'e}cnica Federico Santa Mar{\'i}a\\
    and \\
    Scott Ward\\
    Department of Mathematics, Imperial College London \\
    and \\
    Ed Cohen\\
    Department of Mathematics, Imperial College London \\
   and\\
    Rasmus Waagepetersen\\
    Department of Mathematical Sciences, Aalborg University \\
    }
  \maketitle
} \fi

\bigskip
\begin{abstract}
The $K$-function is a fundamental summary statistic for assessing
clustering or regularity of point processes in two or three
dimensional Euclidean
space. In practice, however, many planar point patterns arise
from projecting locations of objects on  a surface in three
dimensional space to two dimensional space.  For example, when events or
objects occur in a landscape their elevation is often ignored. This can lead to erroneous conclusions regarding properties of the point
process generating the point pattern.

There is not a unique way to
extend the classical $K$-function to point patterns on a complex surface. In this paper we propose, explore, and discuss
several approaches in terms of their theoretical and
computational properties. The best performing approach, coined the surface area $K$-function, can be viewed is an
analogue of the classical $K$-function replacing
counts of points in Euclidean balls with counts of points in surface geodesic
balls. However, an important distinction is that the argument of our surface area
$K$-function is area instead of radius of geodesic balls. The
  performances of the various surface $K$-functions are
  compared in applications
to simulated and real data.
\end{abstract}

\noindent%

\spacingset{1.45}

\section{Introduction}

The main bulk of statistical methodology for point processes is
confined to point processes in $d$-dimensional Euclidean space,
$\R^d$, $d \ge 1$. Recent work has, however, considered point
processes on a variety of non-Euclidean spaces, including spheres
\citep{robeson2014point,moeller2016functional,cuevas2018log}, convex
manifolds
\citep{cuevas2018log,ward2021testing,begu2024nonparametric,ward2025multitype,clemente2026nonparametric}, linear networks \citep{ang2012geometrically,baddeley2021analysing} and arbitrary smooth manifolds for intensity estimation \citep{ward2023nonparametric}. These examples show that the geometry of the domain is intrinsic to the statistical analysis, as it determines the relevant notions of distance, neighbourhood and scale. In this paper we are mainly interested in point processes on surfaces representing the topography of a geographical region, that is, surfaces that can be viewed as graphs of functions mapping $\R^2$ into $\R$. We also outline how our methods can be extended to more general surfaces. 

As a motivating example, we consider point patterns of locations of trees in the Sinharaja plot. The Sinharaja topography shown in the left panel of Figure~\ref{fig:sinharaja} is far from flat, featuring a valley and a maximal elevation difference of $150{\cdot}32$ m over the $500 \times 500$ m plot. Common practice when studying clustering of plant species in such a plot is nevertheless to ignore the third dimension and simply consider the planar point pattern of the two dimensional coordinates of the trees, as in the right panel of Fig.~\ref{fig:sinharaja}.

\begin{figure}
\centerline{%
\includegraphics[width=0.48\textwidth]{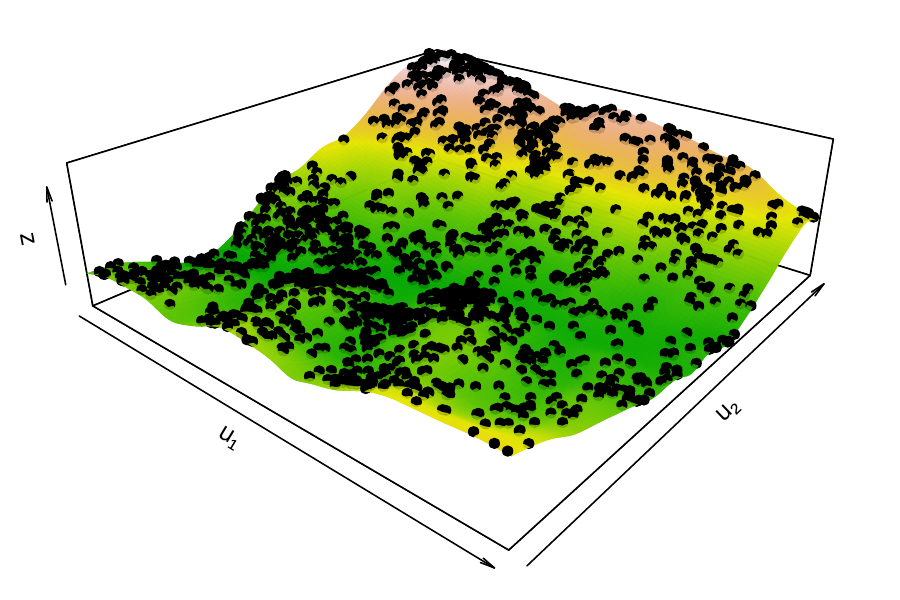}%
\hfill
\includegraphics[width=0.44\textwidth]{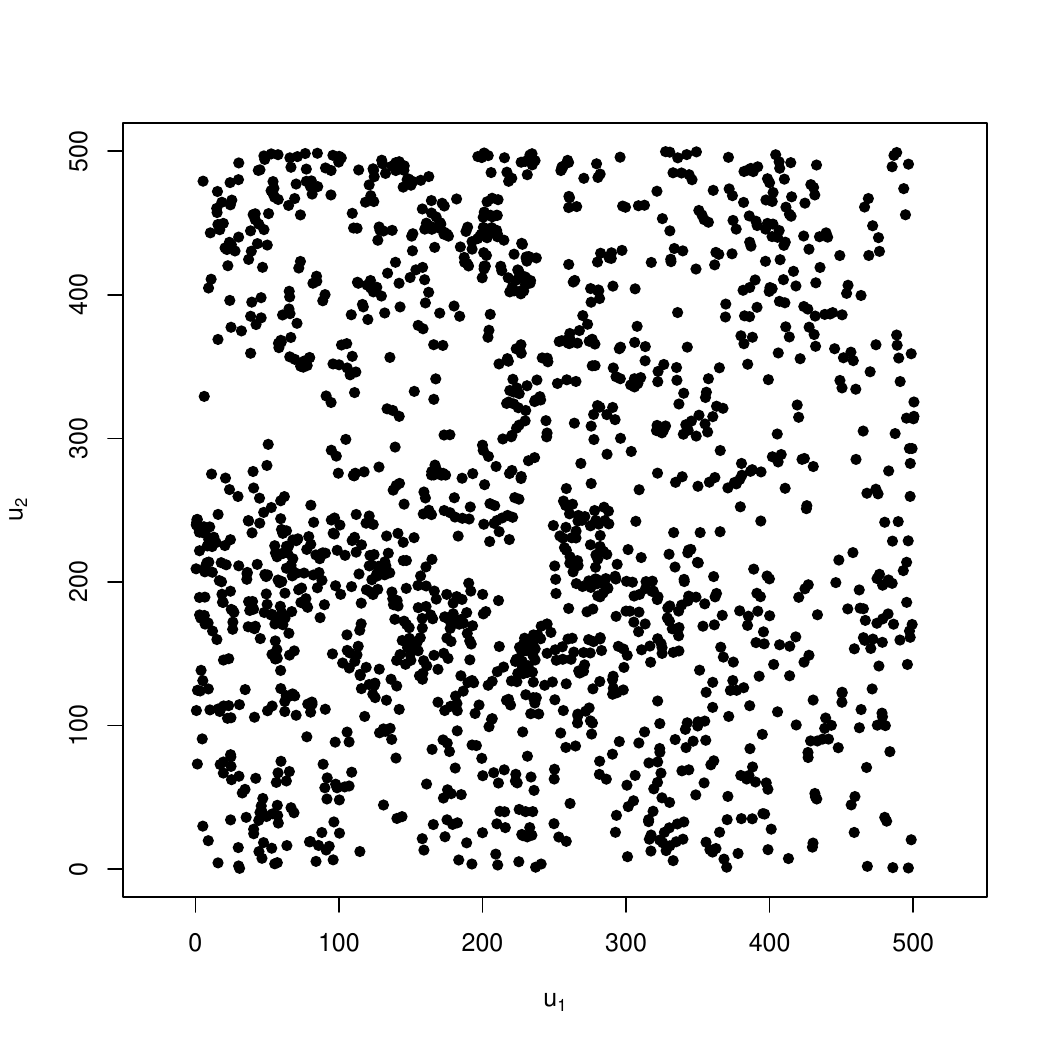}}
\caption{Elevation over the Sinharaja plot with locations of \textit{Anisophyllea cinnamomoides} imposed (left) and two dimensional coordinates of the same locations (right).}
\label{fig:sinharaja}
\end{figure}

\citet{mcdowall2017importance} emphasised that naively analysing spatial patterns obtained by removing the third coordinate may lead to bias in estimating relations between the intensity of points and explanatory variables, and suggested an appropriate adjustment of the intensity function of the projected point process. The projection of points on a surface could also lead to spurious second-order inferences in the projected process. Our work is focused on this latter issue, where applications of the usual $K$-function for planar point processes can lead to misleading conclusions regarding clustering or inhibition properties of the original point process on the three dimensional surface. We therefore propose and compare various surface analogues of the $K$-function that cater for the specificities of the particular surface on which the points occur. In doing so, we clarify which aspects of the Euclidean formulation reflect general principles and which depend specifically on planar geometry, thereby explaining both why naive application of planar methods to projected data can fail and why the proposed surface analogues have different inferential interpretations.

Importantly, these analogues all agree with the classical $K$-function in the case of a flat surface, but on curved surfaces with non constant curvature they differ in geometric meaning as well as computational difficulty. In particular, our comparisons show that the surface area $K$-function provides a natural formulation in which the argument is neighbourhood area rather than radius. This suggests that the classical $K$-function may be viewed most naturally as a function of area, with radius being an equivalent parameterisation only for flat Euclidean geometry and constant curvature surfaces such as the sphere. Code is available as an
  \texttt{R} package at \verb+https://github.com/FcoCuevas87/PointProcess_Surfaces+.

\section{Background and notation for point processes on surfaces}

\subsection{Surfaces and geodesics}\label{sec:surfgeod}

For specificity and ease of presentation we consider parametric surfaces of intrinsic dimension two embedded in three-dimensional Euclidian space. Thus, we define a surface $S \subset \R^3$ as 
\[ S=\{ s \in \mathbb{R}^{3}: s = \varphi(u) \mid u \in D \} \]
where $D \subseteq \R^2$ is an open subset of the plane, $\varphi : \mathbb{R}^{2} \mapsto \mathbb{R}^{3}$ is an injective vector valued function $\varphi(u) = (\varphi_{1}(u), \varphi_{2}(u), \varphi_{3}(u))$, and $u = (u_{1}, u_{2})$ is of dimension two. 

Assume that for almost all the locations the derivative
$\partial_{u_{i}} \varphi(u)=
(\partial_{u_{i}}\varphi_{j}(u))_{j=1}^3$ exists. Then, for all
locations $s=\varphi(u) \in S$, the normal vector $n(u)$ exists. For
the  surfaces considered in this paper, the normal vector is given by $n(u):= \partial_{u_{1}}\varphi(u) \times \partial_{u_{2}}\varphi(u)$. Let $A$ be a subset of $S$ and let $B$ be the preimage of $A$ under $\varphi$. Then, the surface area of a set $A \subseteq S$ is 
\[ \nu(A)= \int_{A} \nu(\dd s) = \int_{B} \| n(u) \| \dd u.\]

Let $s$ and $t$ be two points in $S$ such that $s = \varphi(u)$ and $t = \varphi(v)$ with $u,v \in D$. We define the geodesic distance between $s$ and $t$ as
\[ d_{S}(s,t) = \inf_{\gamma \in \Gamma} \int_{0}^{1}  \| \gamma^{\prime}(\tau) \| {\rm{d}\tau}. \]
Here $\Gamma : \{ \gamma(\tau): [0,1]
\rightarrow S \mid \gamma(\tau) = \varphi(\ell(\tau)), \ell:
  [0,1] \rightarrow D\text{ differentiable}, \ell(0)=u, \ell(1) = v\}$ is the set of smooth
curves that connect $s$ and $t$. Further, $\| \gamma'(\tau) \| =
\sqrt{\ell'(\tau)^\top G( \ell(\tau)) \ell'(\tau)}$ where the $ij$th
entry in the $2 \times 2$ matrix  $G(u)$ is $\langle \partial_{u_{i}}
\varphi(u), \partial_{u_{j}} \varphi(u) \rangle = \sum_{k = 1}^{3} \partial_{u_{i}}
\varphi_{k}(u) \partial_{u_{j}}
\varphi_{k}(u)$, for 
$i,j=1,2$.

We define the geodesic ball $b_{S}(s,r)$ with center $s$ and radius $r$ as the set of points
$t$ in $S$ such that $d_{S}(s,t) \leq r$. Let $a_{s}(r) = \nu(b_{S}(s,r))$ be the area of the geodesic ball with
center $s$ and radius $r$. Except when $S$ is a hyper plane, $a_{s}(r)$ in general differs from $\pi r^2$. For any point $s$ in $S$,
$a_{s}(r)$ is a smooth function of $r$ such that $\dd a_{s}(r)/\dd r > 0$. Then for each $s \in S$ and fixed value of area $a \ge 0$ there exists a radius $r^{*} := r^{*}(s,a)$ such that $a_{s}(r^{*}) = a$. We can therefore also parametrize geodesic balls in terms of their area and denote by $b_s(a)=b_S(s,r^*(s,a))$ the geodesic ball centered at $s$ and of area $a$.

The framework above obviously covers topographic surfaces where
$\varphi(u)=(u_{1},u_{2},z(u_{1},u_{2}))$ and $z:D \rightarrow \R$ is
differentiable. However, other surfaces are covered too, such as
ellipsoids with semi-axes $A$, $B$, and $C$, for which $\varphi(u) = (A \cos(u_{1}) \sin(u_{2}), B \sin(u_{1}) \sin(u_{2}),
  C\cos(u_{2})),$ where $(u_{1}, u_{2}) \in D= [0,
\pi] \times [0, 2\pi]$. The sphere of radius $R$ is the special case $A=B=C=R$.

\subsection{Point processes on surfaces}

We denote by $X$ a point process on a surface $S$ equipped with geodesic distance as defined in the previous section. That is, $X$ is a random locally finite subset of $S$, meaning that for any $A \subseteq S$ bounded with respect to geodesic distance, the cardinality $N(A)$ of $X \cap A$ is finite almost surely. A Poisson process $X$ on $S$ is characterized by the counts $N(A)$ being Poisson distributed (which implies that counts in disjoint subsets of $S$ are independent).

If functions $\ld: S \rightarrow [0,\infty[$ and $g: S \times S \rightarrow [0,\infty[$ exist such that for non-negative functions $h:S \rightarrow [0,\infty[$ or $h: S \times S \rightarrow [0,\infty[$, 
\[ \EE \sum_{u \in X}h(u) = \int_Sh(u)\ld(u) \nu (\dd u) \]
and
\[ \EE \sum_{u,v \in X}^{\neq} h(u,v) = \int_{S}\int_{S}h(u,v) \ld(u)\ld(v) g(u,v) \nu(\dd u) \nu(\dd v) \]
(the so-called first and second order Campbell formulas) then $\ld$
and $g$ are called the intensity function and the pair correlation
function of $X$, respectively. If a regression model is proposed for
$\ld$ this can easily be fitted
using the Poisson likelihood score or a logistic regression estimating
function (Appendix~\ref{sec:estimationintensity}).

Let $\varphi^{-1}: S \mapsto D$ denote the inverse of $\varphi$. For example, in case of the topographical surface, $\varphi^{-1}$ is the function that removes the third coordinate of the point process $X$. Then we obtain a point process on $D$ as the inverse mapping $X_\pl=\varphi^{-1}(X)$ which has intensity function $\ld_\pl(u)=\lambda( \varphi(u) ) \|n(u)\|$ and pair correlation function $g_\pl(u, v)= g( \varphi(u), \varphi(v) )$.
 If $X$ is a Poisson process on $S$, then by the mapping theorem \citep{kingman1993poisson}, $X_D$ becomes a Poisson process on $D$.
 
\section{$K$-functions for a point process on a surface}\label{sec:surfaceK}

The $K$-function, originally defined for point processes on Euclidean space \citep{ripley1977modelling,baddeley2000non}, is a fundamental tool for assessing deviations from
the Poisson process reference model. Consider a point process $Y$ on
$\R^2$ with intensity function $\rho(\cdot)$. Then, the $K$-function for  $Y$ is defined by 
\begin{equation}\label{eq:Kstandard} K(r)  = \frac{1}{|B|} \EE \left [
    \sum_{u \in Y \cap B,v \in
    Y\setminus u } \frac{1[\| u-v \|< r]}{\rho(u) \rho(v)} \right]\end{equation}
for any $B \subset \R^2$ of finite
positive area $|B|>0$, {\em provided} the right hand side does not depend on
the choice of $B$. If the latter requirement is fulfilled, $Y$
is said to be second order intensity reweighted stationary (SOIRS)
\citep{baddeley2000non}. A sufficient condition for SOIRS is that the
pair correlation function for $Y$ is translation invariant, $g(u,v)=g(0,v-u)$ in which case
\[ K(r)=\int_{\|v\|\le r} g(0,v-u) \dd v \]
for any $u \in \R^2$. In the case of constant intensity $\rho(u)=c$, $u \in \R^2$, $c K(r)$ is the expected number of further points from $Y$ within distance $r$ from $u$ given that $Y$ already has a point at $u$. This explains why the $K$-function is a useful summary of clustering of $Y$.

Extending the usual $K$-function to point processes on surfaces is not
straightforward since concepts like stationarity or translation
invariance are only available for special cases of non-flat surfaces like for
example a sphere where one can consider invariance under
rotations. Rather than attempting to extend the notion of SOIRS to
general surfaces we initially pursue a more modest goal of defining empirical analogues of
the $K$-function that make intuitive sense as measures of clustering
or inhibition and have a known form under the important null hypothesis of an
inhomogeneous Poisson process. Comments on extensions to non-Poisson point processes are provided in Section~\ref{sec:non-Poisson}.

Specifically, we consider three alternative $K$-functions that we coin: surface corrected $K$-function, tangent plane $K$-function, and surface area $K$-function. The tangent plane $K$-function turns out to be computationally very demanding but is conceptually a natural stepping stone between the two other alternatives. For completeness, we include a brief account of the tangent plane $K$-function here and provide further details in the supplementary material. We do not, however, consider it in the simulation study or data examples.

\subsection{Surface corrected  $K$-function}\label{sec:corrected}

In case of a three-dimensional surface that can be
bijectively mapped to a sphere, \cite{ward2021testing} explored  that Poisson processes
that are transformed by a bijective mapping remain Poisson. Hence
they transformed a point process on a three-dimensional surface to a
point process on a sphere which enabled them to take advantage of the
inhomogeneous $K$-function for spheres \citep{moeller2016functional} to test the null-hypothesis that
the observed point pattern originates from a Poisson process.

A similar approach can be used in our setting. Suppose we observe $X$ inside $S_W=\varphi(W)$ for some $W \subseteq D$. Then, we can simply apply the standard planar $K$-function to $X_{\pl}$ to obtain a `surface corrected' planar $K$-function 
\begin{equation}\label{eq:Ksc} \hat K_{\text{sc}}(r)=  \frac{1}{|W|} \sum_{u \in X_\pl \cap W}\sum_{v \in X_\pl
      \setminus u} \frac{1[ \|u-v\| \le r]}{\lambda_{\pl}(u)\lambda_{\pl}(v)}
\end{equation}
which has known expected value $\EE \hat K_{\text{sc}}(r)=\pi r^2$ when $X$ and hence $X_D$ are Poisson processes \citep{baddeley2000non}. In practice edge corrections are introduced to address the issue that  $v$ in $X_D$ outside $W$ are not observed.

As previously mentioned it is common practice in case of topographic surfaces to apply the usual
$K$-function to $X_{\pl}$ but without applying the correction of the
intensity function given by the product of surface area elements
$\|n(\varphi(u))\| \|n(\varphi(v))\|$. That is, an intensity function is applied that may be
systematically biased due to the lack of surface adjustment. This can,
as demonstrated in the simulation study (Section~\ref{sec:simulation}),
lead to false conclusions regarding clustering properties of $X$.

An issue with applying the surface corrected planar $K$-function \eqref{eq:Ksc} is that the images in $S$ of discs in the
plane may take very different and sometimes non-appealing forms on the
surface. For example, considering a topographic surface, a disc could be transformed back into a very elongated
ellipsoid if the surface is steep. Such a neighbourhood may not seem
reasonable for the original point process $X$ due to the strong anisotropy that may not have any
interpretation in relation to the point process. This problem can be
mitigated by projecting the points in $X$ to tangent planes of $S$
 as detailed in the next section.

\subsection{Tangent plane $K$-function}\label{sec:tangentplane}

Consider a point $s$ on the surface, and let $T_s$ denote the tangent
plane of $S$ at $s$. Assume that $h_{ST}$ is a mapping (e.g.\
projection) from $S$ to $T_s$. We restrict the domain of $h_{ST}$ to a
connected neighbourhood $N_s \subset S_W$ of $s$ so that $h_{ST}$
becomes injective as a mapping from $N_s$ to $T_s$. We choose $N_s$ to
be the maximal such neighbourhood. The injectivity of $h_{ST}$ on
$N_s$ is crucial for derivations in the supplementary material.

We then let $X_{T_s} \subset A_{T_s}= h_{ST}(N_s)$ be the point process of points $h_{ST}(t)$, $t \in X \cap N_s$ and let $b_T(s,r) \subset T_s$ be the disc with
centre $s$ and radius $r$ contained in $T_s$.  Moreover, denote by
$\ld_{T_s}$ the intensity function of $X_{T_s}$ (derived in supplementary
Section~\ref{supp-sec:tangentplaneintensity}).
 We then define a tangent plane $K$-function as 
\begin{equation}\label{eq:tangentK} \hat K_{\text{tangent}}(r)=\frac{1}{\nu(S_W)} \sum_{s \in X \cap
    S_W} \sum_{v \in X_{T_s}\setminus h_{ST}(s)} \frac{1[ v \in b_T(s,r)]}{\ld(s) \ld_{T_s}(v)}e_T(s,r). \end{equation}
 In this definition, $\ld(s)=\ld_{T_s}(s)$,  and $e_T(s,r)= \pi r^2/\nu( b_T(s,r) \cap A_{T_s})$
is  a correction factor. For a Poisson process, $\EE \hat
K_{\text{tangent}}(r) = \pi r^2$ (see supplementary Section~\ref{supp-sec:intro}).

A simple choice for $h_{ST}$ would be the orthonal projection of $s$
onto $T_s$. One issue here is that the backtransform $h_{ST}^{-1}(b(s,r))$ of
a ball in $T_s$ is in general not a geodesic ball of radius $r$ in
$S$. One might instead use the so-called logarithm  map whose
inverse, the exponential map,  has the appealing property of
transforming Euclidean balls of a given radius on the tangent space
into geodesic balls of the same radius on $S$ \citep{doCarmo1992}. Also in general the
injectivity radius (supplementary Section~\ref{supp-sec:intro}) of the exponential mapping, will be greater than the
injectivity radius for the orthogonal projection. For a sphere of
radius $R$ for example, the injectivity radius is $\pi R/2$ for the orthogonal
projection and $\pi R$ for the exponential mapping. Details of the tangent
$K$-function with logarithm and exponential maps are given in the
supplementary Section~\ref{supp-sec:riemannian}.

While the tangent plane $K$-function is appealing as a natural
refinement of the surface corrected $K$-function, its practical use is
hampered by computational obstacles. The orthogonal projection and its
Jacobian needed to evaluate $\ld_{T_s}$ (supplementary Section~\ref{supp-sec:jacobian}) are easy to compute but the computation of $A_{T_s}$ is difficult. In case of the logarithm mapping, the computation of $A_{T_s}$ is less of a problem but the computation of the Jacobian is very time consuming for general surfaces. This has led us to pursue in the next section a third alternative relying on geodesic balls of fixed area rather than fixed radius.

\subsection{Surface area $K$-function}\label{sec:area}

As mentioned in the previous section, the tangent plane $K$-function
using the logarithm map essentially considers neighbouring points in
geodesic balls on the surface of a given radius $r$. Indeed, it would
be tempting to directly define a $K$-function for $X$ mimicking the
$K$-function \eqref{eq:Kstandard} for Euclidean spaces, replacing
Euclidean distance with geodesic distance. However, the expected value
of this would, even in case of a Poisson process, not be independent
of the choice of $B \subset S$. This is because for a general surface,
surface areas of geodesic balls with a given radius will in general
vary depending on the location of the center of the geodesic ball. Observe, however, that although the usual $K$-function is
parametrized by a length, its value in case of a Poisson process is the
area of a ball of radius $r$. More generally, neighbourhood area
  is the natural reference
  point when assessing counts in
  neighbourhoods. This leads us to consider geodesic balls of a given
surface area rather than a given radius.

Recall from Section~\ref{sec:surfgeod} that for $s \in S$ and $a>0$,
$b_s(a)$ denotes the geodesic ball with center $s \in S$ and surface area $a$. Then, for a point pattern observed in $S_{W} \subset S$, we define 
\begin{equation}\label{eq:saK} \hat  K_{\text{area}}(a)=
  \frac{1}{\nu(S_W)} \sum_{s \in X \cap S_W}\sum_{t \in (X \setminus
    s) \cap S_W} \frac{1[ t \in b_s(a)]}{\lambda(s)\lambda(t)}e_{S}(s,a), \end{equation}
where $e_{S}(s,a) = a/\nu\left(b_s(a) \cap S_W\right)$ is an edge correction factor with $e_{S}(s,a) = 1$ if $b_s(a) \subset S_{W}$ and $e_{S}(s,a) > 1$ if $b_s(a) \cap S_{W} \neq b_{s}(a)$. With this, the area $K$ function has expectation  $\EE \hat K_{area}(a)=a$ in case of a Poisson
process. 

Although the definition of this surface area $K$-function is simple, the numerical implementation is not straightforward. However, for any location $s$ on $S$, geodesic distances to all other locations in $S$ can be determined by solving the so-called Eikonal equation, see Appendix~\ref{sec:eikonal}. In practice this is done numerically for a fine grid of locations in $S$. Having determined the geodesic distances to the grid locations, approximations of geodesic balls $b_S(s,r)$ for any radius $r$ can be obtained in terms of the grid points that lie within distance $r$ from $s$, and so, the area of $b_S(s,r)$ can be approximated numerically via grid integration. Hence, the $r^*(s,a)$ is determined for a given $a$ using a grid search.

\section{Non-Poisson case}\label{sec:non-Poisson}

All the proposed $K$-functions have known expectations given by $\pi r^2$ or area $a$ in case of a Poisson process $X$. For more general $X$ one may impose conditions on the pair correlation function ensuring that the defined $K$-functions still have expectations not depending on the particular observation window $W$ considered. As discussed in the beginning of Section~\ref{sec:surfaceK}, for a planar point process it suffices that the pair correlation function is translation invariant. This in turn means that for the surface corrected $K$-function \eqref{eq:Ksc},
\[ \EE \hat K_{\text{sc}}(r) = \int_{b(0,r)} g_{\pl}(0,v) \dd v \]
provided $g_{\pl}(u,v)=g(\varphi(u), \varphi(v))$ only depends on $v-u$ (translation invariance) or on $\|v-u\|$ (isotropy).

For a point process on the surface $S$ it may seem more natural
to impose
conditions in terms of geodesic distance and following
\cite{rakshit2017},  we may define $X$ to be $d_S$-correlated if
$g(s,t)=g_0(d_S(s,t))$ is a function of $d_S(s,t)$, for all $s,t \in
S$. In this case one can show (supplementary Section~\ref{supp-sec:riemannian}) that the tangent plane $K$-function \eqref{eq:tangentK} implemented with the logarithm mapping satisfies
\[ \EE \hat K_{\text{tangent}}(r) = 2 \pi \int_0^r g_0(\tau) \dd \tau .\]

Unfortunately it is difficult for general surfaces $S$ to construct explicit models that satisfies $g(s,t)$ is a function of $\varphi^{-1}(s)-\varphi^{-1}(t)$ or of $d_S(s,t)$. Spheres are exceptions where construction of rotation invariant models with $g(s,t)=g_0(d_S(s,t))$ are quite easy but this setting is limited with our scope of application in mind.

For the surface area $K$-function \eqref{eq:saK} we have by the Campbell formula that
\[ \EE \hat K_{\text{area}} (a) = \frac{1}{\nu(S_W)} \int_{S_W} \int_{b_s(a) \cap S_W} g(s,t) e_{S}(s,a) \dd s \dd t \]
and we would require that $ \int_{b_s(a)}g(s,t) \dd s$ does not depend on $s$. Since this depends on an interplay between the surface within $b_s(a)$ and the pair correlation function $g(s,t)$, $t \in b_s(a)$, a simple condition on the pair correlation function does not seem obvious.

\section{Simulation experiments}\label{sec:simulation}

In Sections~\ref{sec:corrected}-\ref{sec:area} we introduce three proposals for $K$-functions adapted to non-flat surfaces: surface corrected \eqref{eq:Ksc}, tangent plane \eqref{eq:tangentK}, and surface area $K$-function \eqref{eq:saK}. In practice the choice among these alternatives depend on both computational feasibility and their power for detecting deviations from complete spatial randomness (meaning that the point process $X$ is a Poisson process). It turns out that the numerical implementation of the tangent plane $K$-function is computationally very heavy due to the need for evaluating the set $A_{T_u}$. We therefore restrict attention to the surface corrected $K$-function which is computationally straightforward and the surface area $K$-function for which we could obtain a computationally efficient implementation based on the Eikonal equation. For comparison we also include the naive approach where we simply evaluate the ordinary planar $K$-function for $X_{\pl}$. For the simulation experiments we consider three types of surfaces and three different types of point process models.

\subsection{Surfaces for simulation experiments}\label{sec:surfaces}

The three surfaces considered are defined parametrically over a common domain $D= [-1,1] \times [-1,1]$. The \textit{Monkey Saddle} surface is given by
$
    \varphi(u) = \bigl(u_{1},\, u_{2},\, u_{1}^3 - 3u_{1}u_{2}^2\bigr),
$
which presents a saddle point at the origin with three upward and three downward directions, resulting in highly variable curvature. The \textit{Gaussian Bump} surface is defined as
$
    \varphi(u) = \bigl(u_{1},\, u_{2},\, 5\exp(-(u_{1}^2+u_{2}^2)/0.5)\bigr),
$
which features a smooth, symmetric elevation concentrated at the origin with rapidly decaying 
curvature toward the boundary. Finally, the \textit{Polynomial Hills} surface is given by
$
    \varphi(u) = \left(u_{1},\, u_{2},\, \frac{0.15}{0.02+(u_{1}-1)^2+0.1u_{2}^2} + 
    \frac{0.15}{0.02+(u_{1}+1)^2+0.1u_{2}^2}\right),
$
which presents two distinct peaks of equal height located symmetrically along the $u_{1}$-axis, with a saddle-like region between them. Each surface is depicted in Figure \ref{fig:simulation_open_surfaces}.
\begin{figure}
    \centering
    \includegraphics[width=0.32\linewidth]{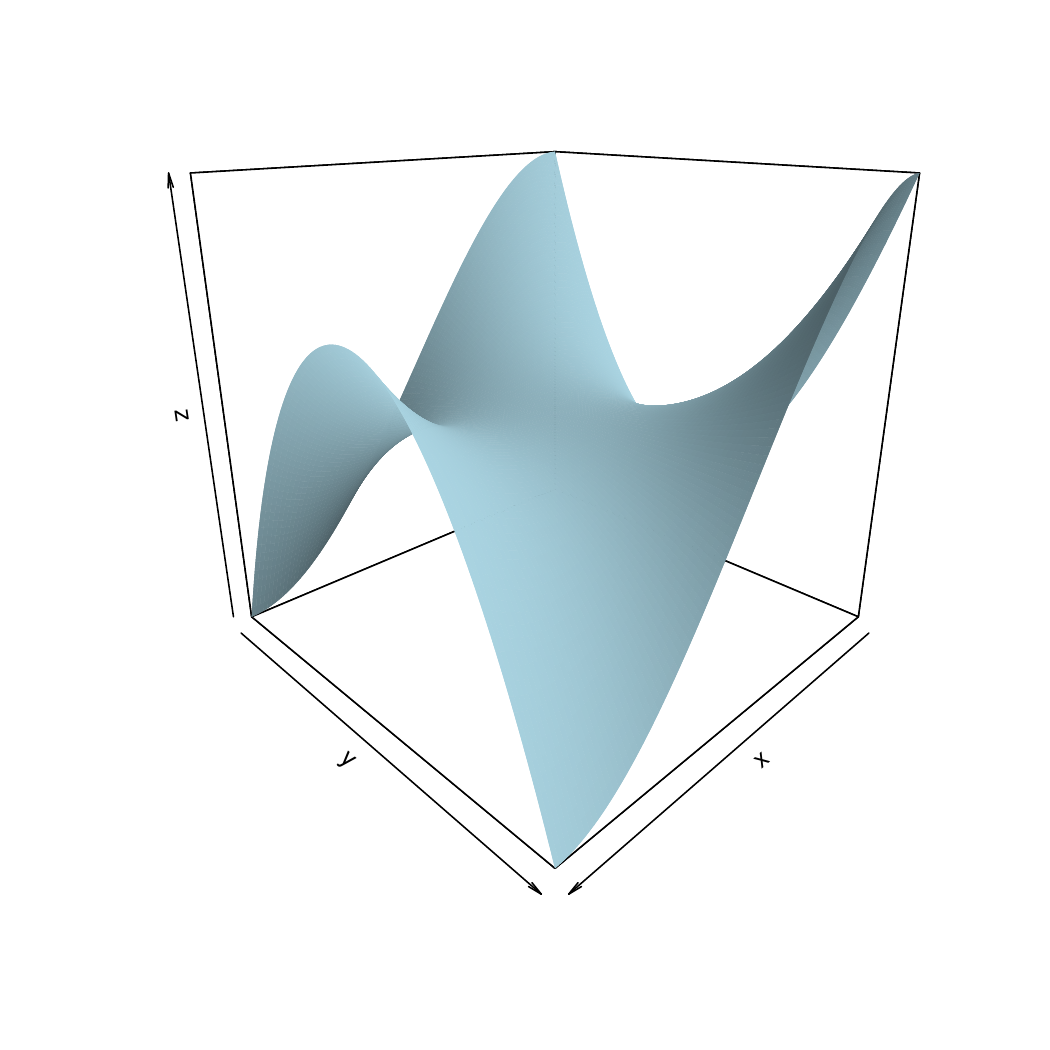}     \includegraphics[width=0.32\linewidth]{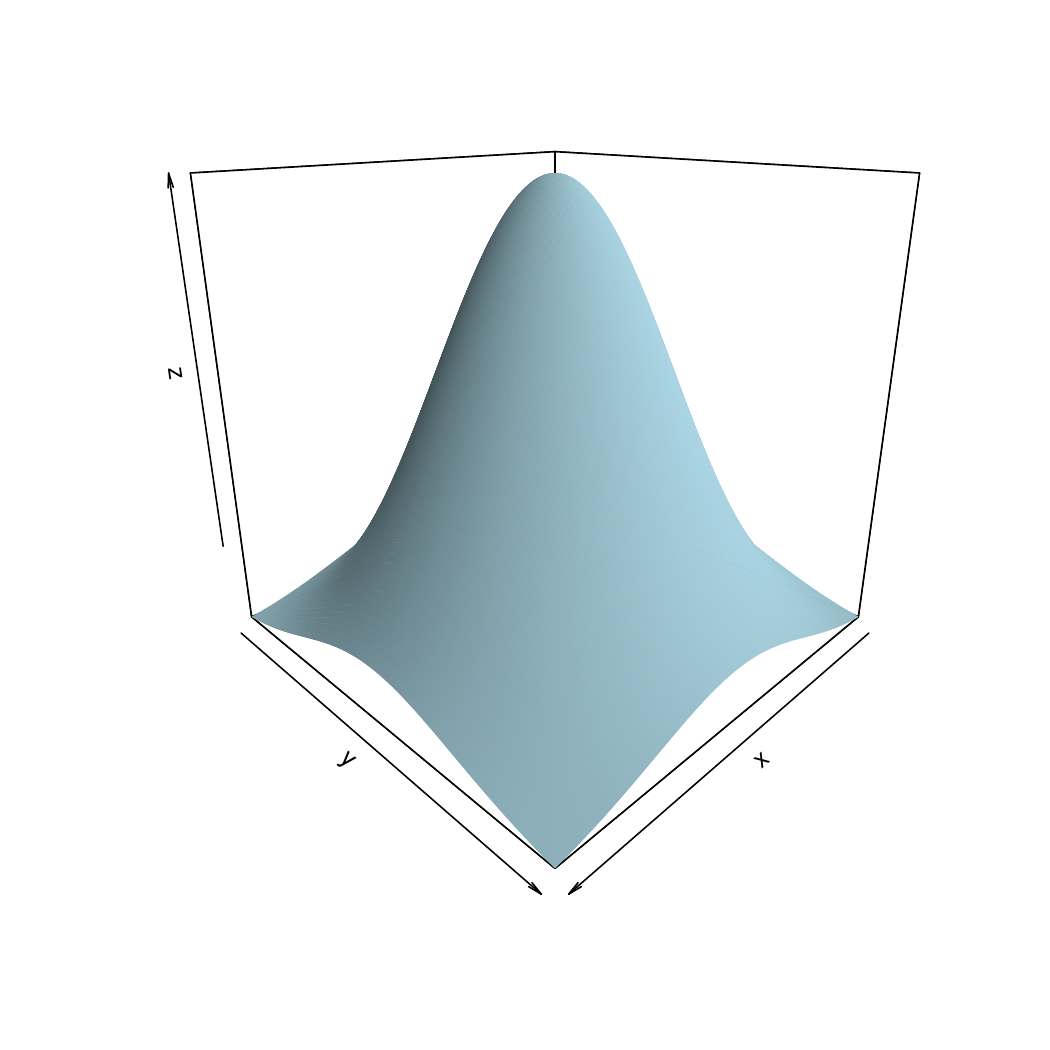}     \includegraphics[width=0.32\linewidth]{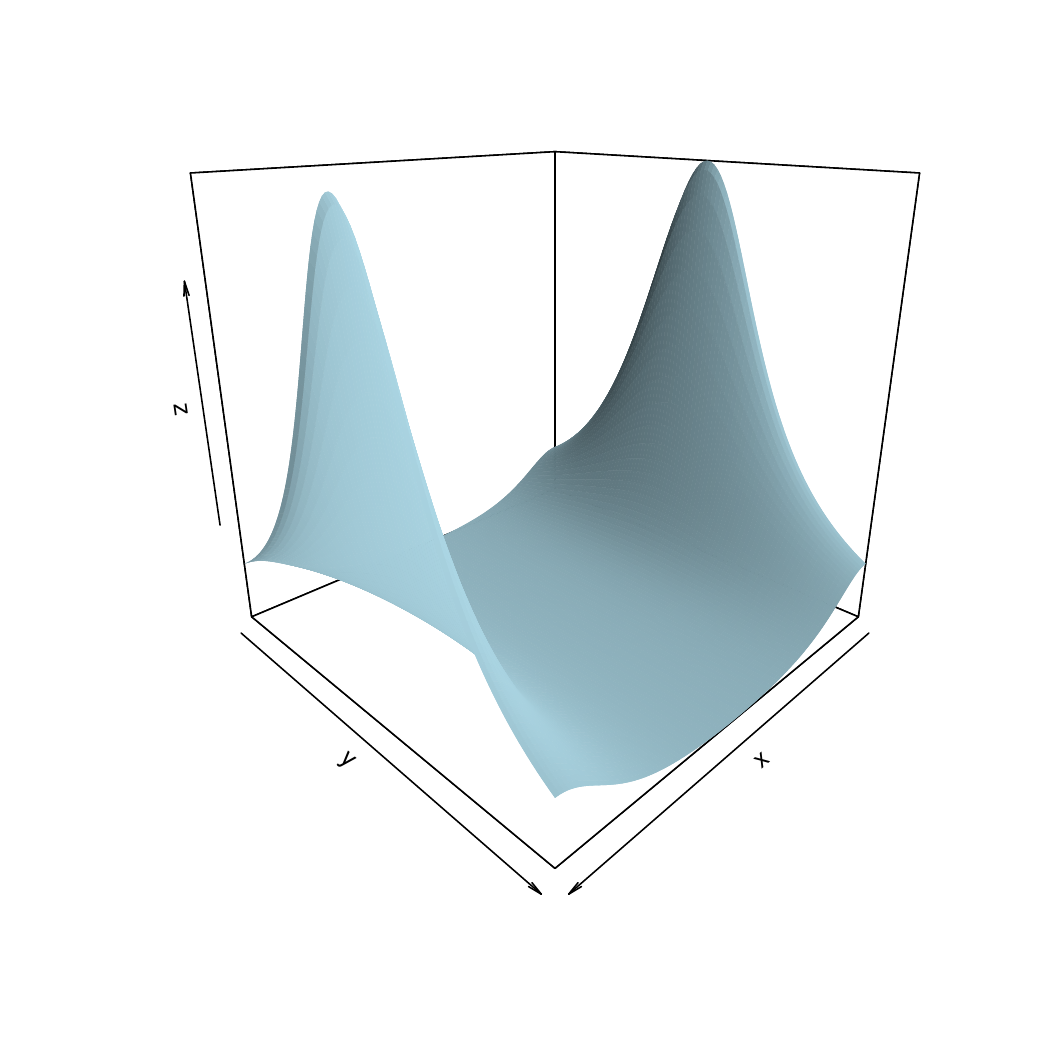}
    \caption{Plot of the different surfaces. Left: Monkey saddle, center: Gaussian bump, right: polynomial hills.}
    \label{fig:simulation_open_surfaces}
\end{figure}

\subsection{Surface point process models }\label{sec:models}

We generate point process realizations from a homogeneous Poisson
process as well as surface analogues of Matérn cluster and
  Matérn type II inhibitory point processes. The homogeneous
Poisson
process of intensity $200/\nu(S)$ is sampled by
first generating a Poisson random variable $N$ with expectation 200,
and next sampling $N$ uniformly distributed points on
$S$ using the algorithm
in \cite{diaconis2013sampling}. For the clustered point process we first
simulate a Poisson point process with intensity $\kappa = 25$, and next for
each point $s$ in the Poisson process, we simulate the
offspring via a Poisson point process with intensity $\mu = 8$ on
$b_{s}(a)$, with $a = 0.2\nu(S) $. Then, the resulting point process
of intensity approximately $200/\nu(S)$ is the union of the
offspring. Finally, to simulate inhibitory point patterns we
initially simulate a Poisson process of intensity $\lambda = 100$ together
with a uniform mark. Next, for each point $s$ of the Poisson
  point process, $s$ is retained if and only if its mark is smaller than
the marks of all its neighbours $t$ satisfying $t \in b_{s}(a)$ or  $s
\in b_{t}(a)$, with $a= 0.0025\nu(S)$. This results in a point
process with intensity $82.8$ for the monkey saddle,
$81.4$ for the Gaussian bump, and $77.2$ for the polynomial hills.

\subsection{Study of power for testing complete spatial randomness}

For each combination of surface and point process model we generate
1000 independent simulated point patterns and compute the three types
of $K$-functions for each simulated pattern. We next for each point
pattern test the complete spatial randomness (CSR) null hypothesis that the points are generated from a homogeneous Poisson process using each of the $K$-functions. This is done using the global rank envelope test (GET) \citep{myllymaki2017global} implemented in the \texttt{R} package \texttt{GET} \citep{myllymaki2024get}. This returns a $p$-value for each $K$-function. We reject CSR at the nominal level 5\% if the $p$-value is less than 5\%. For the GET test we generated 2499 simulations of a homogeneous Poisson process with intensity estimated from the simulated point pattern. The version of the GET test used for the simulation test ignores the effect of estimating the intensity which can lead to actual signifance levels deviating from the nominal level. A double sampling scheme can be used to implement a GET test correcting for this but this is too heavy computationally for the simulation study.

Table~\ref{tab:k_comparison} shows the actual levels of the global envelope tests with each of the three $K$-functions. In case of the Poisson process, the naive $K$-function is essentially useless with very high rejection rates over 87\%. The actual level for the surface corrected $K$-function is somewhat too small due to ignoring the effect of estimating the intensity for each simulated pattern. Interestingly, the results for the surface area $K$-function seem less sensitive in this respect with actual levels quite close to the nominal level for all types of surfaces.

Obviously, the naive $K$-function ignoring the surface is not reliable for testing CSR. Hence for the clustered and inhibitory patterns we restrict attention to the surface corrected and surface area $K$-functions. The surface area $K$-function clearly outperforms the surface corrected $K$-function in terms of power for testing CSR for all types of surfaces. One reason could be that the neighbourhood of each point within the surface makes more physical sense for the surface area $K$-function where the neighbourhoods are geodesic balls.
\begin{table}[ht]
\centering
\begin{tabular}{lccccccc}
\hline
 & \multicolumn{3}{c}{Poisson} & \multicolumn{2}{c}{Clustered} & \multicolumn{2}{c}{Repulsive} \\
\cmidrule(lr){2-4} \cmidrule(lr){5-6} \cmidrule(lr){7-8}
Surface & naive & correc.\ & area & correc.\ & area & correc.\ & area \\
  \hline
  Monkey   & 0.917 &0.025 & 0.049 &  0.337& 0.850  &  0.878& 0.982  \\
  Gaussian   & 0.873 & 0.021 & 0.041 &  0.312 & 0.732 &   0.945 &0.987  \\
Polynomial & 1.000 & 0.028 & 0.039 &  0.969 & 1.000 &   0.674 &0.831  \\
\hline
\end{tabular}
\caption{Empirical rejection rates at nominal significance level $\alpha = 0.05$ for the three 
$K$-functions, surfaces and point process models.}
\label{tab:k_comparison}
\end{table}

\section{Applications}
\subsection{Tree locations in the Sinharaja region}

The data used in this section are gathered in a 25ha  $500m \times
500m$ research  plot W located at $6^{\circ}$ $24^{\prime}N$, $80^{\circ}$
$24^{\prime}E$ in the southwestern part of Sinharaja, Sri Lanka.  The
dataset contains the spatial locations of $237$ different tree
species. Elevation data are available on a dense regular grid
covering the plot. We use a combination of kernel smoothing
and interpolation to obtain a smooth function $\varphi$ representing
elevation for any location in the plot. The highly varying topography of the plot is shown
in Figure~\ref{fig:sinharaja}. The estimated surface area of the plot,
$\hat{\nu}(S)=$28.9ha, is considerably larger than the
planar area of the plot. 

We apply the discussed $K$-functions to point
patterns of tree locations for two species:
\textit{Anisophyllea cinnamomoides} (1653 trees) and \textit{Dalbergia
  pseudo-sissoo} (510 trees), see Figure~\ref{fig:species_all}. 
\begin{figure}[htbp]
    \centering
    \begin{subfigure}[t]{0.45\textwidth}
        \centering
        \includegraphics[width=\textwidth]{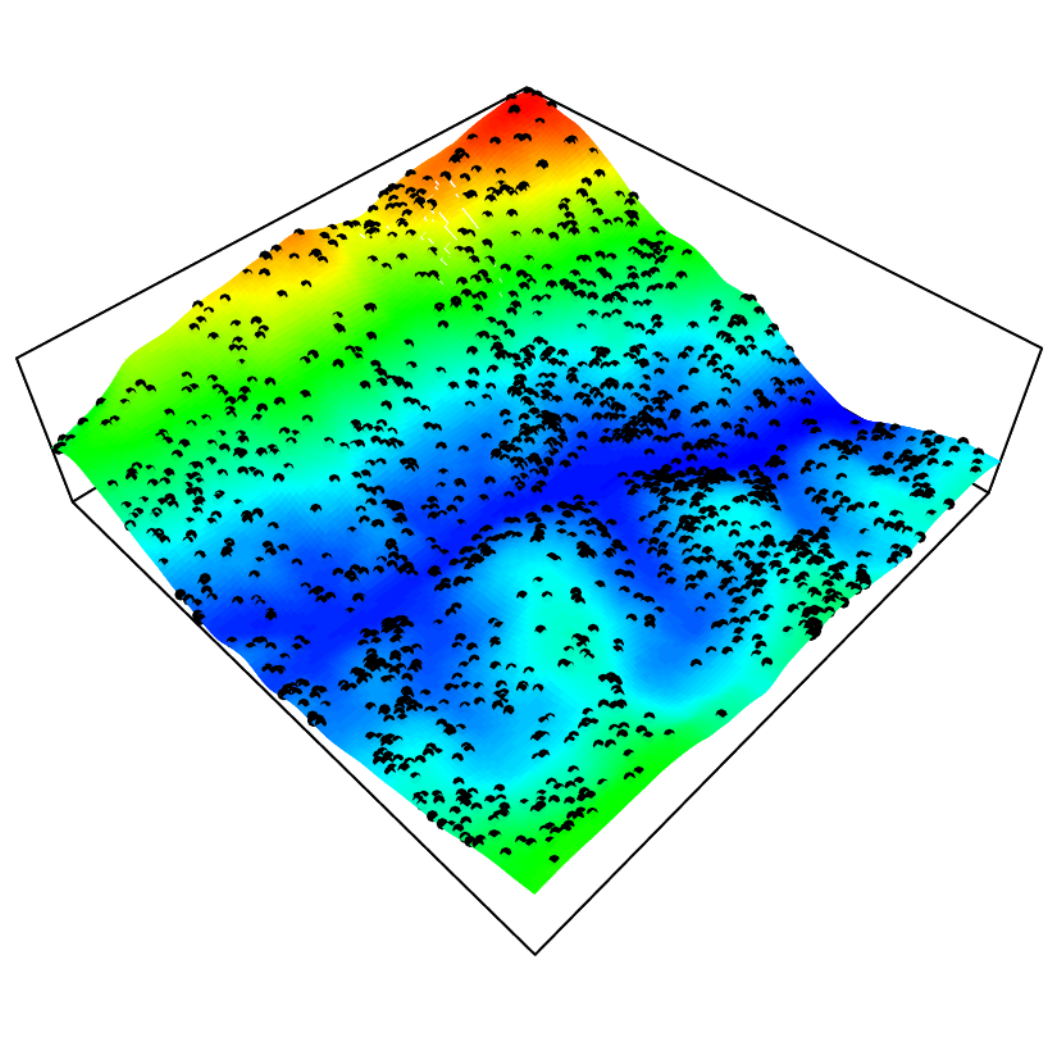}
           \end{subfigure}
    \hfill
    \begin{subfigure}[t]{0.45\textwidth}
        \centering
        \includegraphics[width=\textwidth]{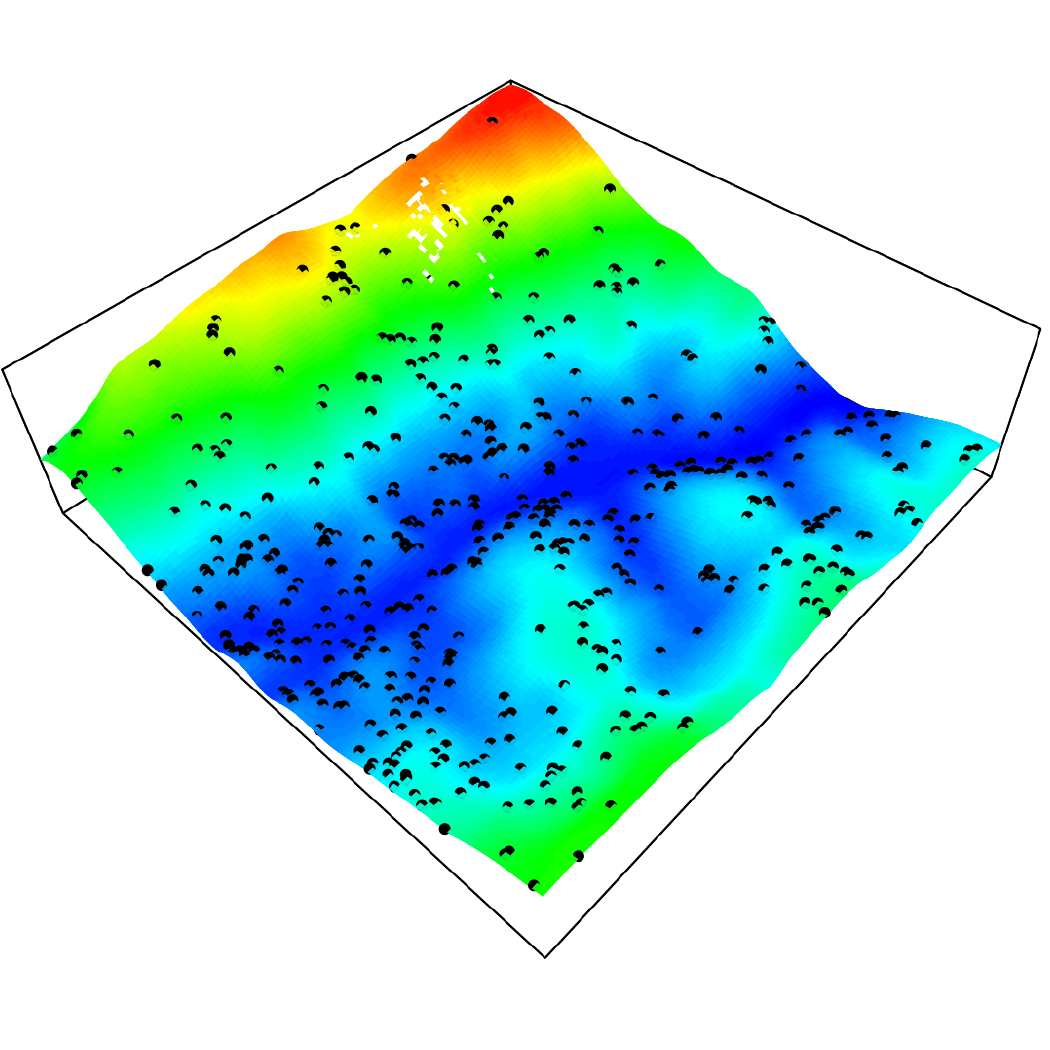}
            \end{subfigure}
    \caption{Patterns of tree locations for two tree
      species (color scale represents elevation). Left: \textit{Anisophyllea cinnamomoides}. Right: \textit{Dalbergia pseudo-sissoo}.}\label{fig:species_all}
\end{figure}

We do not have access to environmental covariates such as soil properties
but one might expect that the intensities of the trees could depend on
topographic features such as elevation and slope. For
each species, we hence estimate (following Appendix~\ref{sec:estimationintensity}) the surface intensity function using a parametric 
log-linear model of the form 
\begin{equation}\label{eq:dataintensity}
\log\lambda((u,z(u))) = \beta_0 + \beta_1 z(u) + \beta_2 u_1 + \beta_3 u_2 
+ \beta_4 \partial_{u_1} z(u) + \beta_5 \partial_{u_2} z(u),
\end{equation} where $z(u) = \varphi_3(u)$ denotes the elevation at location
$u = (u_1, u_2) \in W$, and $\partial_{u_1} z(u)$, $\partial_{u_2} z(u)$ 
are the terrain gradients in each parametric direction. 
The fitted intensity functions are shown in the lower row of
Figure~\ref{fig:intensity_all}. 
\begin{figure}[htbp]
    \centering
\begin{subfigure}[t]{0.45\textwidth}
        \centering
        \includegraphics[width=\textwidth]{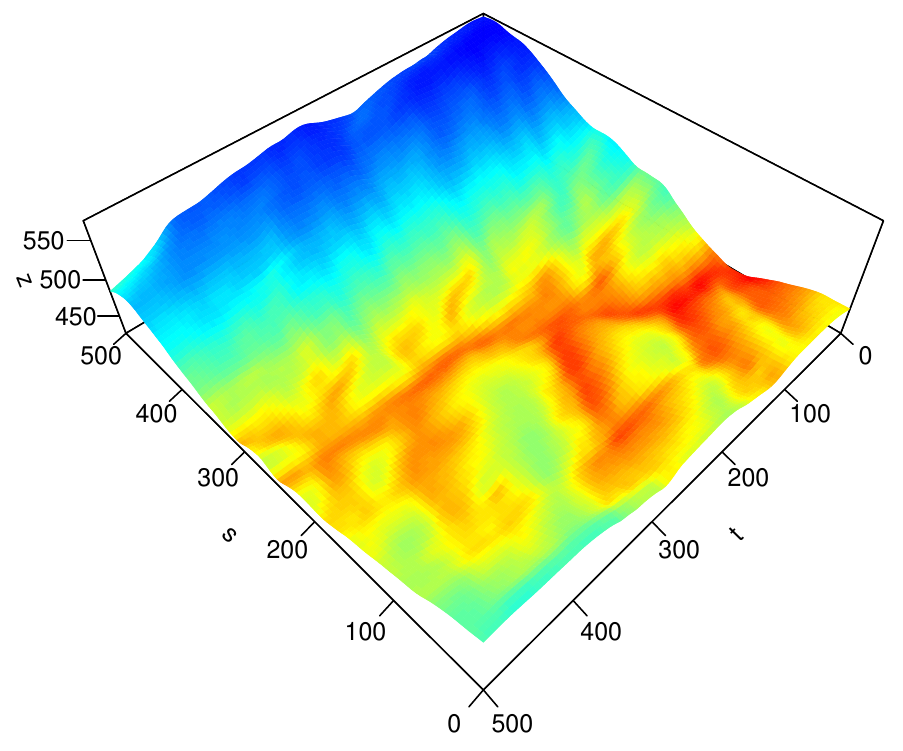}
     \end{subfigure}
     \hfill
    \begin{subfigure}[t]{0.45\textwidth}
        \centering
        \includegraphics[width=\textwidth]{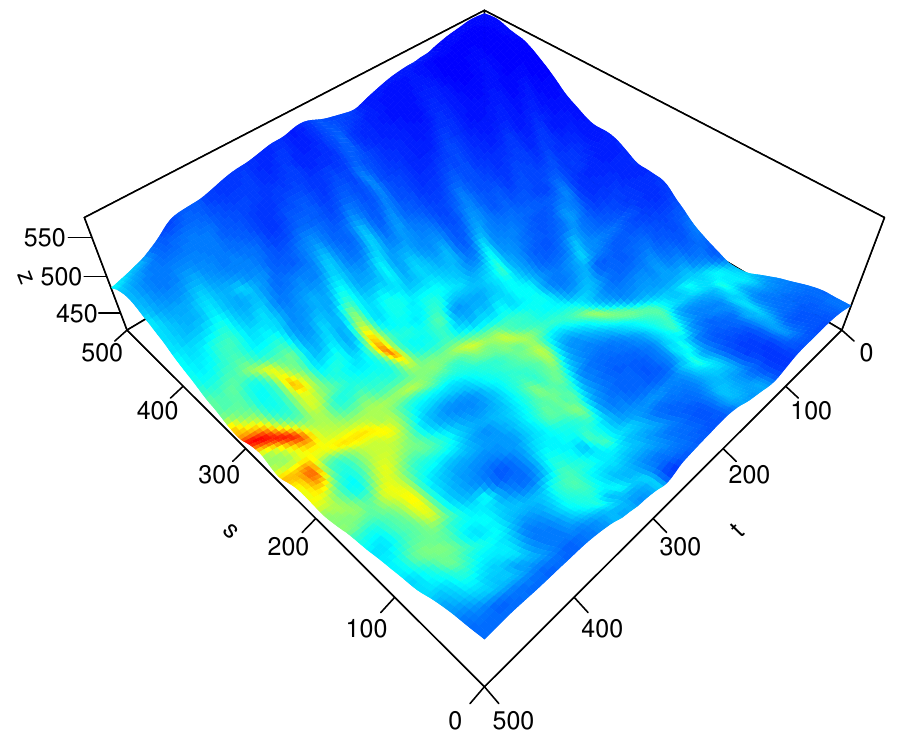}
     \end{subfigure}
    \caption{Estimated surface log-linear intensity
      function for the two species.  Left: \textit{Anisophyllea cinnamomoides}. Right: \textit{Dalbergia pseudo-sissoo}.}\label{fig:intensity_all}
\end{figure}

The estimated intensity functions are
plugged in for the true surface intensity functions in the
surface corrected and surface area $K$-functions \eqref{eq:Ksc} and
\eqref{eq:saK}. For comparison we also compute the standard
inhomogeneous planar $K$-function \eqref{eq:Kstandard} with the log
intensity function $\rho(u)$ also specified by the right hand side
of \eqref{eq:dataintensity}, but crucially now viewed as an intensity
function defined on the planar region $D$. This ignores that trees live on a surface in
three-dimensional space and corresponds to the current practice when
using a $K$-function to analyse clustering of rain forest trees.

We test the CSR null hypothesis via global envelope tests \citep{myllymaki2017global,myllymaki2024get}  based 
on $2499$ Monte Carlo simulations of inhomogeneous Poisson 
processes with the fitted intensity functions, either on the surface for the surface corrected and surface
area $K$-functions, or in the plane for the standard inhomogeneous $K$-function. 
The global envelope tests are shown in Figures~\ref{fig:Kfunctions:Aniso} and~\ref{fig:Kfunctions:Dalb}. In all
cases we plot $K$-functions as a function of area rather than
distance. We moreover center the $K$-functions by subtracting area meaning that the reference value (dashed line in the plots) for the centered $K$-functions
is zero under the Poisson process.
\begin{figure}[!htbp]
  \centering
  \begin{subfigure}[t]{0.32\textwidth}
    \includegraphics[width=\textwidth]{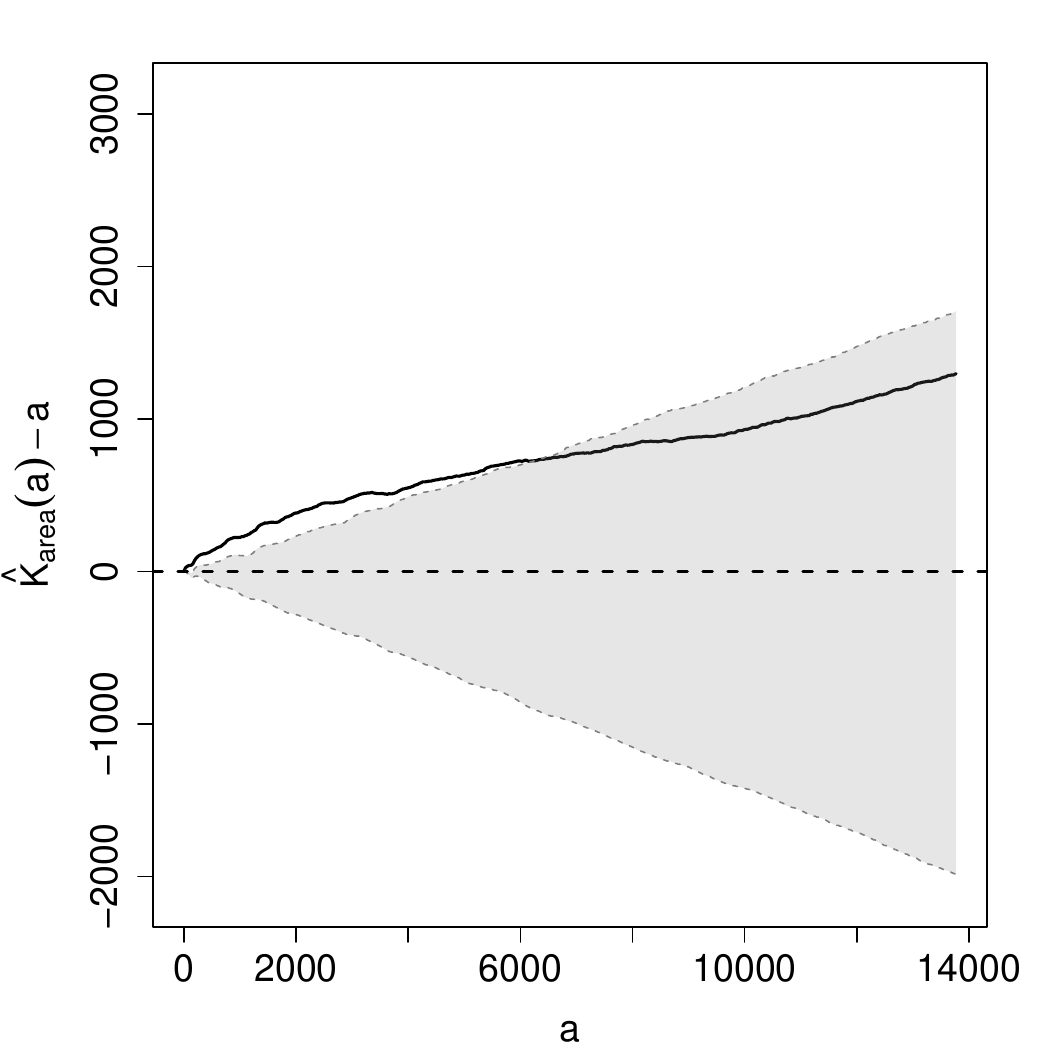}
    \caption{$\hat{K}_{\mathrm{area}}(a)-a$}
    \label{fig:Karea:Aniso}
  \end{subfigure}
  \hfill
  \begin{subfigure}[t]{0.32\textwidth}
    \includegraphics[width=\textwidth]{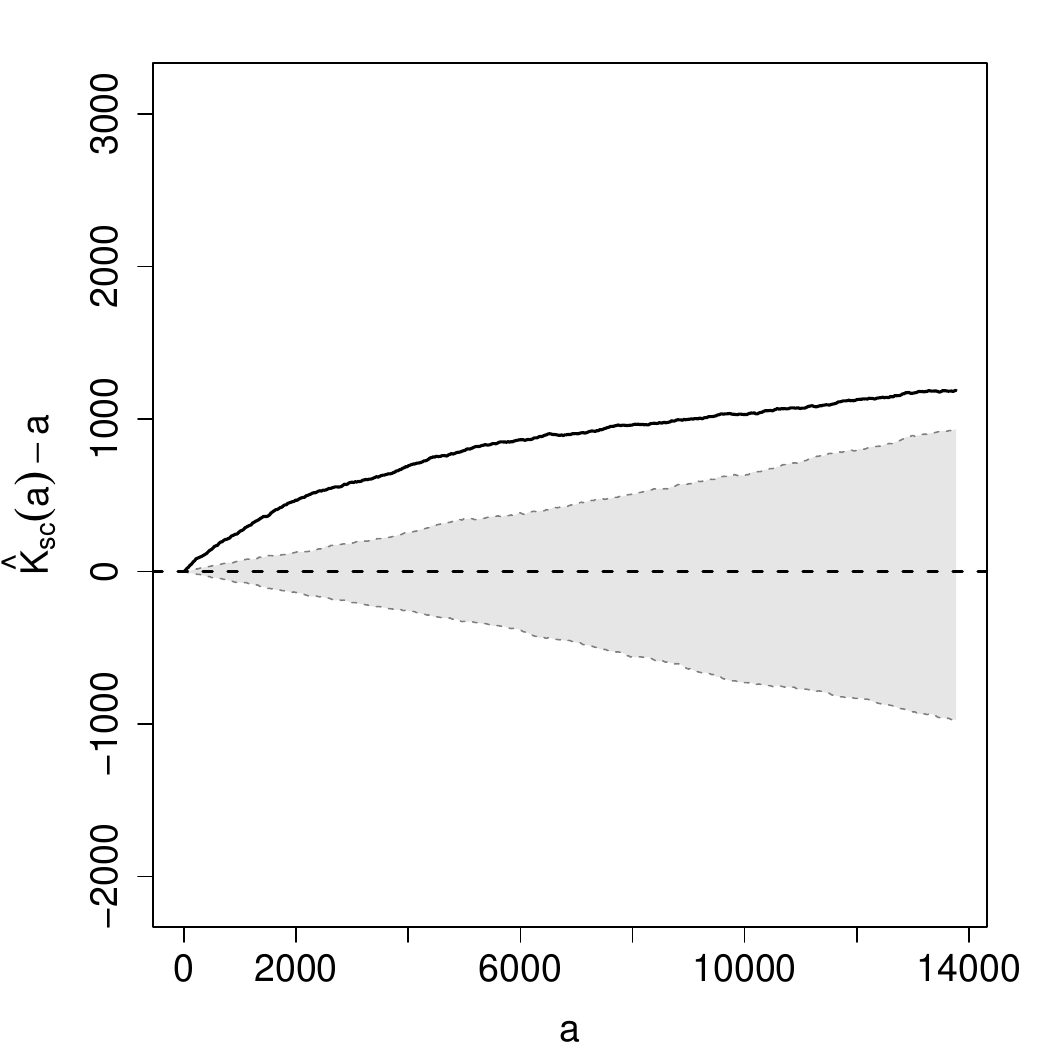}
    \caption{$\hat{K}_{\mathrm{sc}}(a)-a$}
    \label{fig:Ksurfcorr:Aniso}
  \end{subfigure}
  \hfill
  \begin{subfigure}[t]{0.32\textwidth}
    \includegraphics[width=\textwidth]{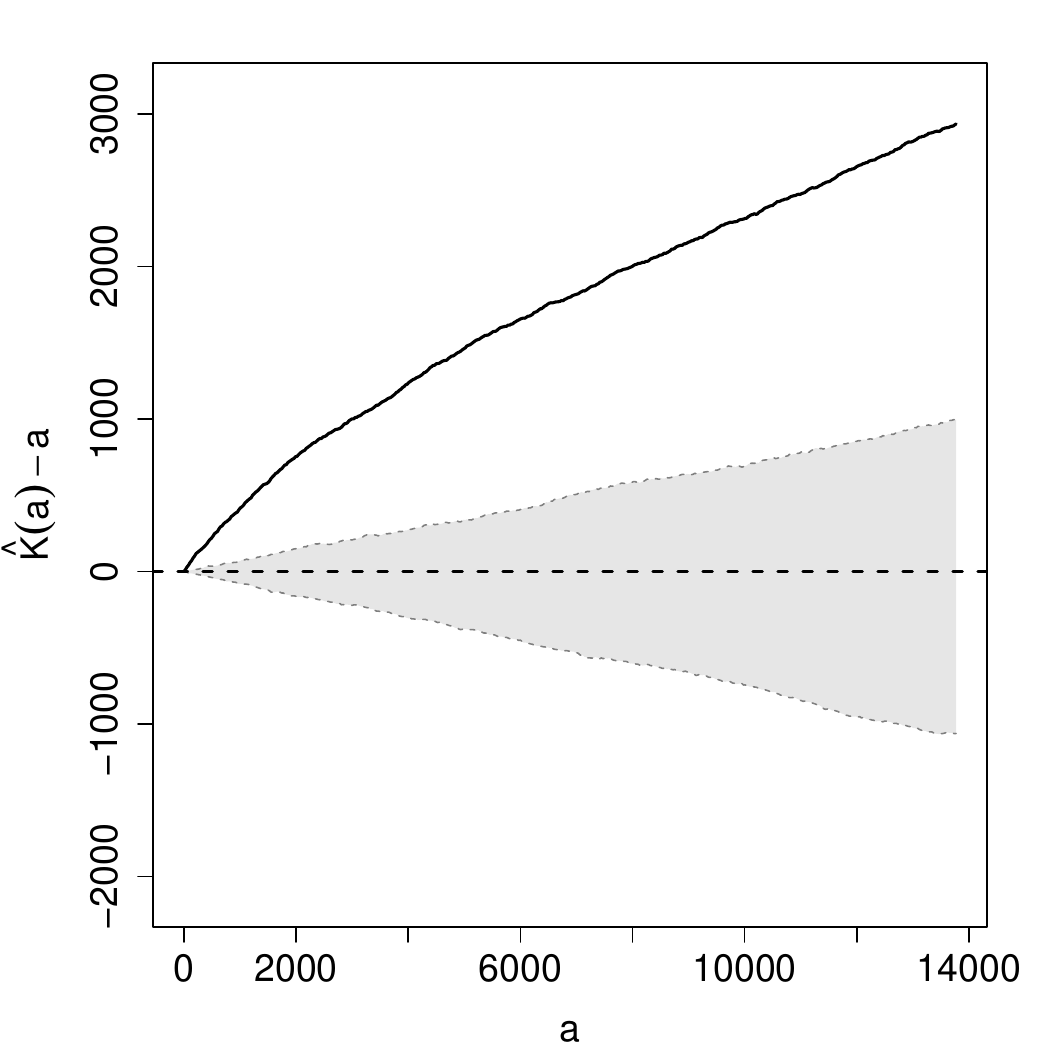}
    \caption{$\hat{K}(a)-a$}
    \label{fig:Kinhom:Aniso}
  \end{subfigure}
  \caption{Global envelope tests for \textit{Anisophyllea cinnamomoides}.
    Grey band: 95\% global envelope under the fitted inhomogeneous Poisson model.
    Black line: observed curve. Dashed line: mean centered $K$-function for Poisson process.}
    \label{fig:Kfunctions:Aniso}
\end{figure}
\begin{figure}[!htbp]
  \centering
  \begin{subfigure}[t]{0.32\textwidth}
    \includegraphics[width=\textwidth]{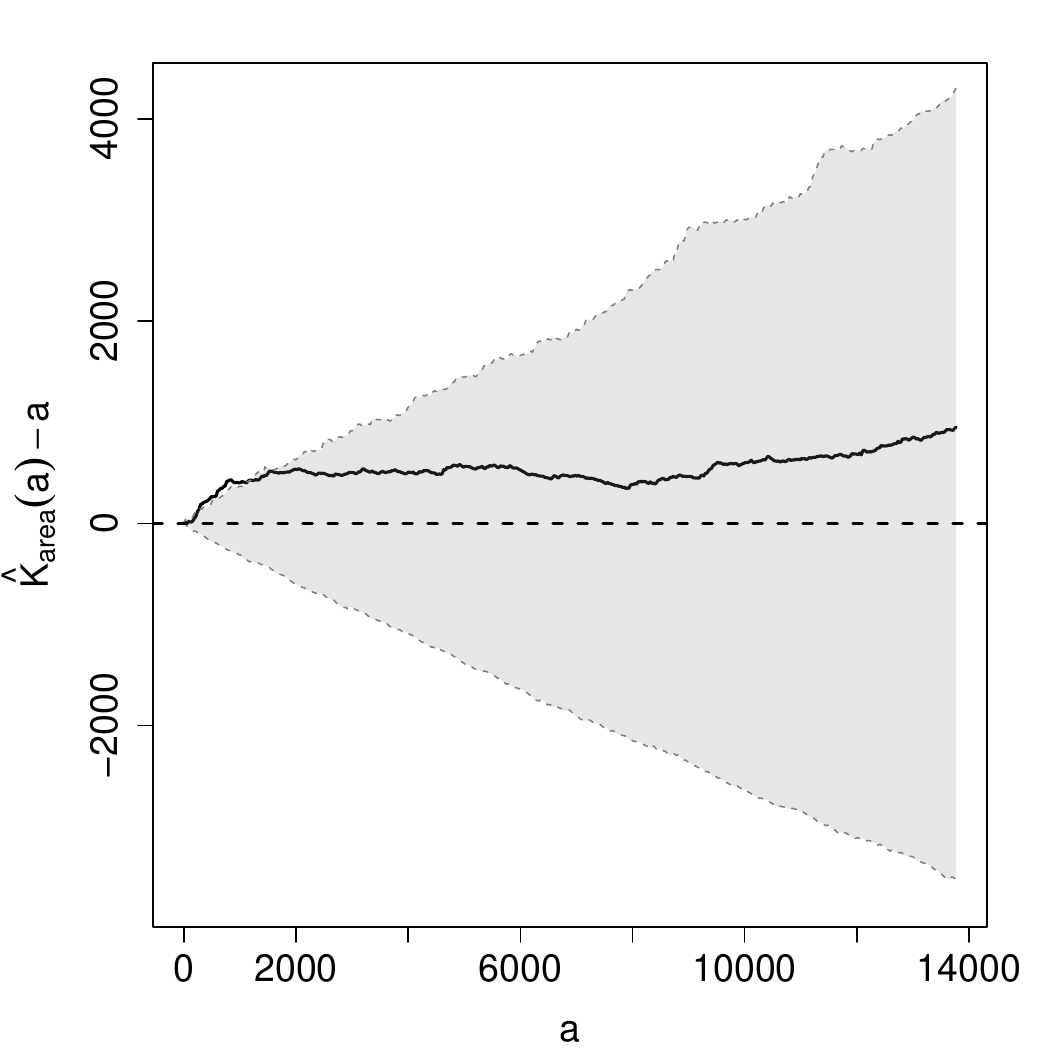}
    \caption{$\hat{K}_{\mathrm{area}}(a)-a$}
    \label{fig:Karea:Dalb}
  \end{subfigure}
  \hfill
  \begin{subfigure}[t]{0.32\textwidth}
    \includegraphics[width=\textwidth]{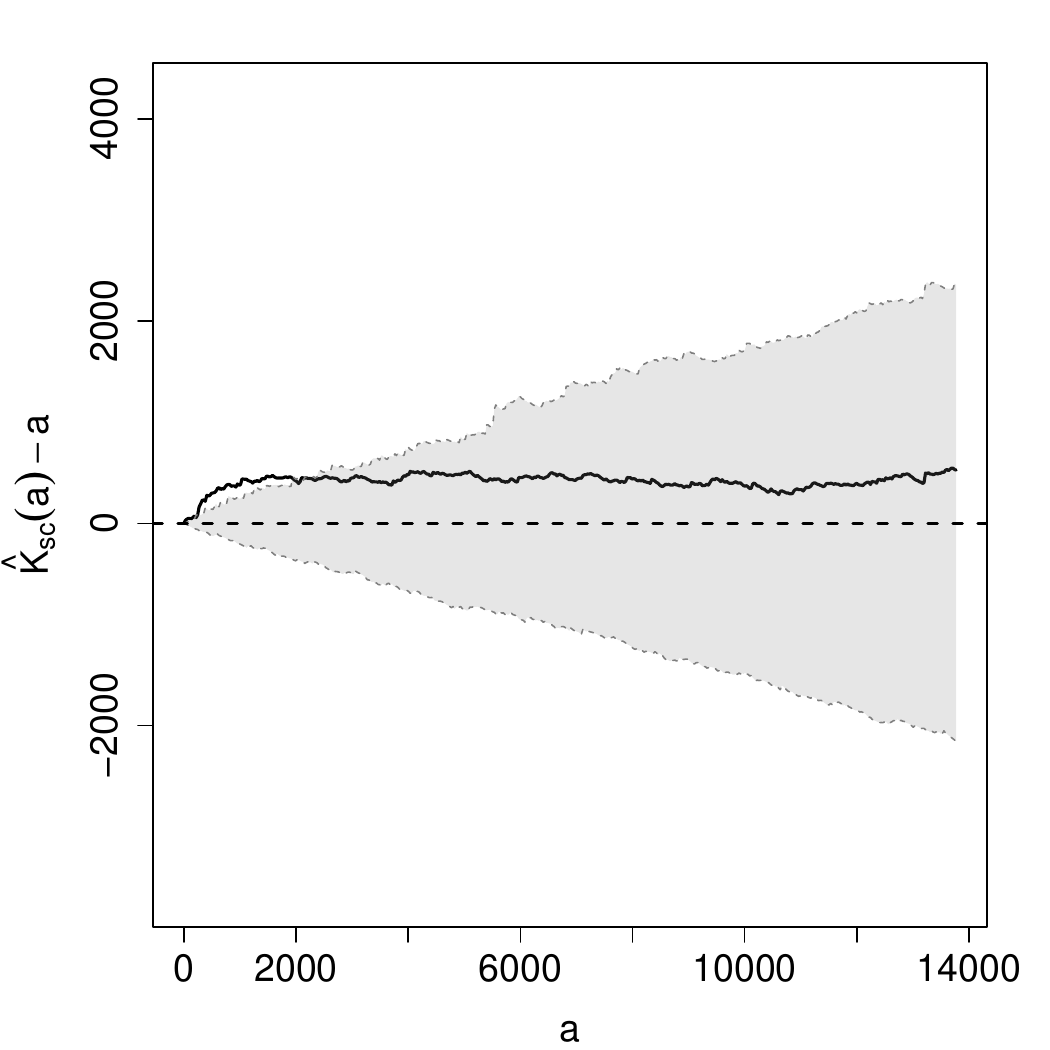}
    \caption{$\hat{K}_{\mathrm{sc}}(a)-a$}
    \label{fig:Ksurfcorr:Dalb}
  \end{subfigure}
  \hfill
  \begin{subfigure}[t]{0.32\textwidth}
    \includegraphics[width=\textwidth]{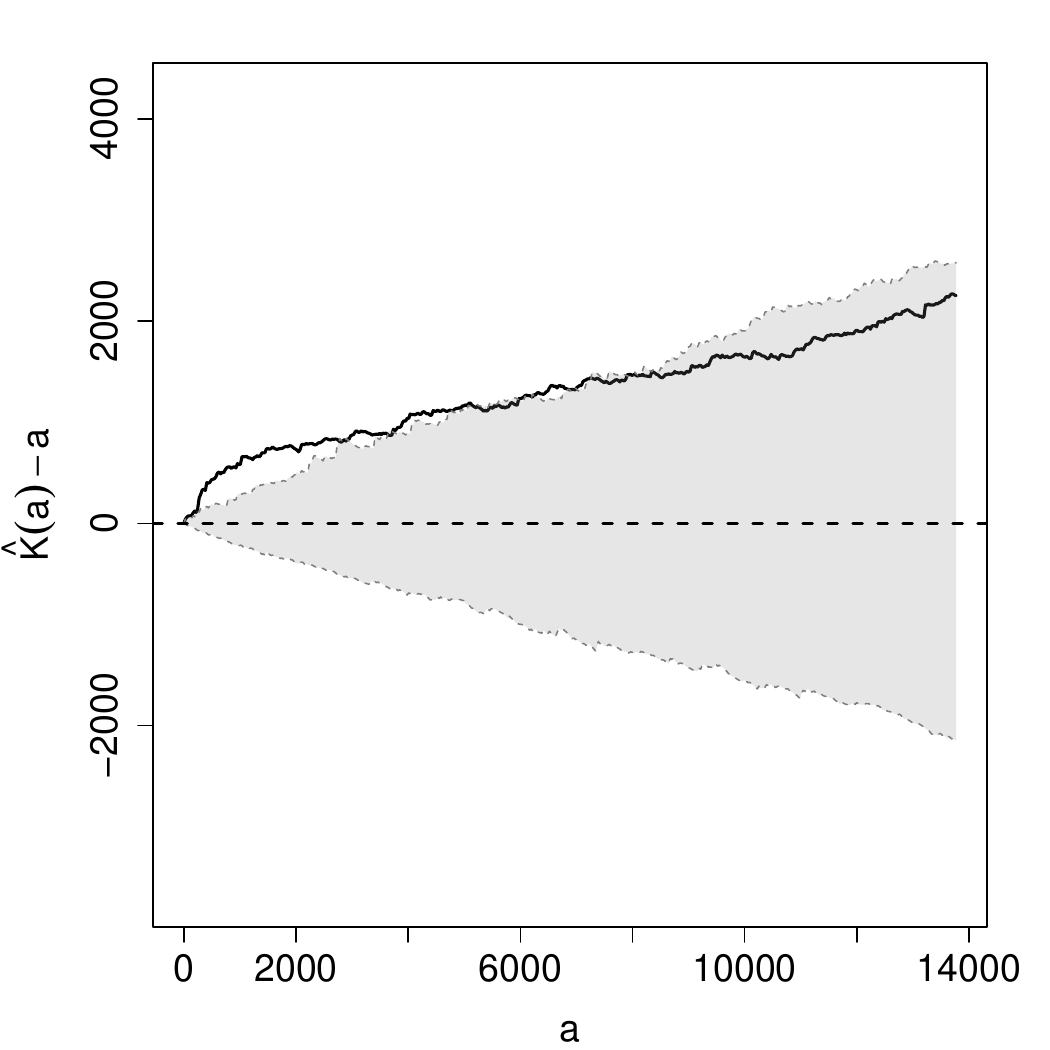}
    \caption{$\hat{K}(a)-a$}
    \label{fig:Kinhom:Dalb}
  \end{subfigure}
  \caption{Global envelope tests for \textit{Dalbergia
      pseudo-sissoo}.  Grey band: 95\% global envelope under the fitted inhomogeneous Poisson model.
    Black line: observed curve. Dashed line: mean centered $K$-function for Poisson process.}
    \label{fig:Kfunctions:Dalb}
\end{figure}

In terms of assessing the null-hypothesis, all of the
$K$-function results qualitatively agree on rejecting at the 5\%
level due to clustering ($K$-functions above the envelopes). However,
quantitatively, the deviation from the null is much bigger for the
standard $K$-function than for the surface-aware $K$-functions. This is likely caused by the standard $K$-function picking
up spurious clustering arising when tree locations on the non-flat
topographic surface are projected to the plane. Thus to get a quantitatively
correct assessment of the level of clustering it is important to account
for the fact that trees exist on a three-dimensional
surface. Also recall that the simulation study in the previous
  section showed lack of type I error control for the standard planar $K$-function.

\subsection{Lipopolysaccharide molecules on the surface of Escherichia coli bacteria}

Lipopolysaccharide (LPS) molecules are a major component of the outer 
membrane of Gram-negative bacteria such as \textit{Escherichia coli}. 
Understanding the spatial organisation of LPS molecules on the bacterial surface is therefore of 
considerable biological interest.

The data \citep{Mamou2025} consist of $2370$ coordinates of LPS molecule locations on 
the surface of a single \textit{E. coli} bacterium, see left plot in 
Figure~\ref{fig:ecoli}. Following
\citet{Kumar2025}, the bacterial surface is approximated 
by a pill-shaped surface  $S$, defined by a cylindrical 
body of half-length $h = 737.4$ nm and hemispherical caps of radius 
$r = 462.6$ nm, fitted by minimising the sum of squared distances from 
the molecule locations to the surface. The estimated surface area of the 
bacterium is $\hat{\nu}(S) = 2\pi r (2r + 2h) = 6.97\mu m^2$.

We apply the surface area $K$-function $\hat{K}_{\mathrm{area}}$ 
\eqref{eq:saK} to the LPS localisation data and test the null
hypothesis of homogeneous CSR using the global envelope
test based on $2499$ Monte Carlo simulations from a homogeneous
Poisson process with expected number of points equal to the observed
number of points. The result is shown in the right plot in
Figure~\ref{fig:ecoli}, where the area $K$-function clearly indicates
rejection of the CSR hypothesis, suggesting that the LPS dataset
exhibits clustering that should be incorporated into any subsequent
downstream modelling. This conclusion is entirely consistent with the
analysis of \cite{Kumar2025}, in which CSR was also rejected using the
projection-based method of \cite{ward2021testing} onto the
  sphere $\mathbb{S}^2$. 

A key advantage of the area $K$-function proposed here is that it is a function of area in the natural units of the bacterial surface, and is therefore directly interpretable in the context of the data. In contrast, the spherical projection approach of \cite{ward2021testing} produces a $K$-function of a radius $r$ defined as great-circle distance on $\mathbb{S}^2$, which does not immediately correspond to a meaningful distance on the bacterium. Furthermore, the estimated area $K$-function appears to lie more clearly outside the envelopes than in the analysis of \cite{Kumar2025}, suggesting stronger evidence against CSR. However, we refrain from making any claims about the comparative power of the two approaches, since this would require a more extensive investigation beyond the scope of the present paper. That said, such a finding would be consistent with \cite{ward2021testing}, who showed for point processes on ellipsoids that the statistical power of their approach decreases as the surface deviates further from $\mathbb{S}^2$.

\begin{figure}[htbp]
    \centering
\begin{tabular}{cc}
    \includegraphics[width=0.4\textwidth,height=0.6\textwidth]{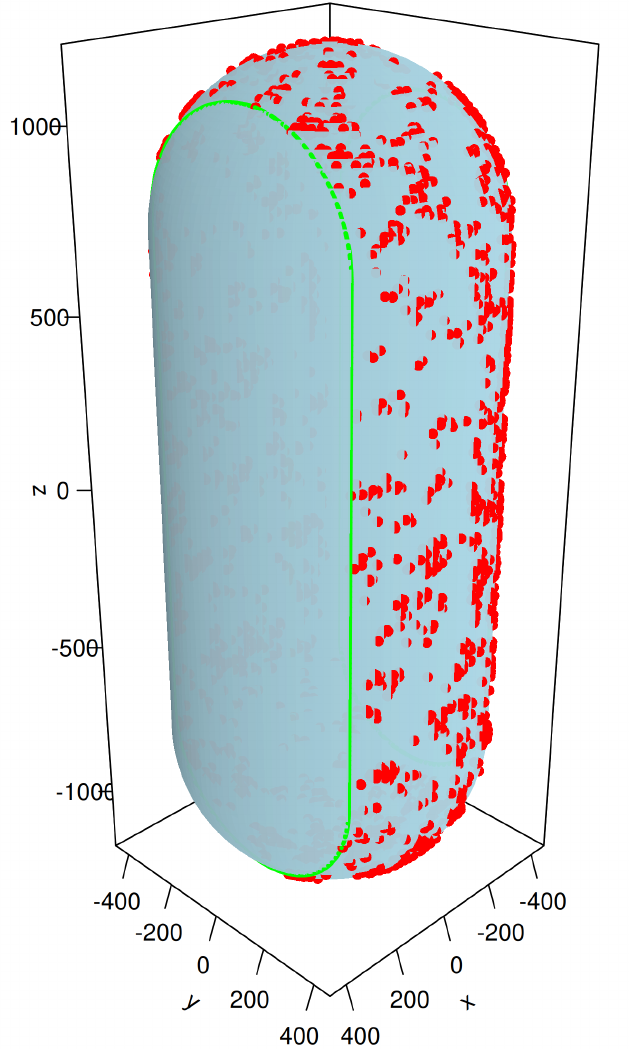} & \includegraphics[width=0.6\textwidth,height=0.6\textwidth]{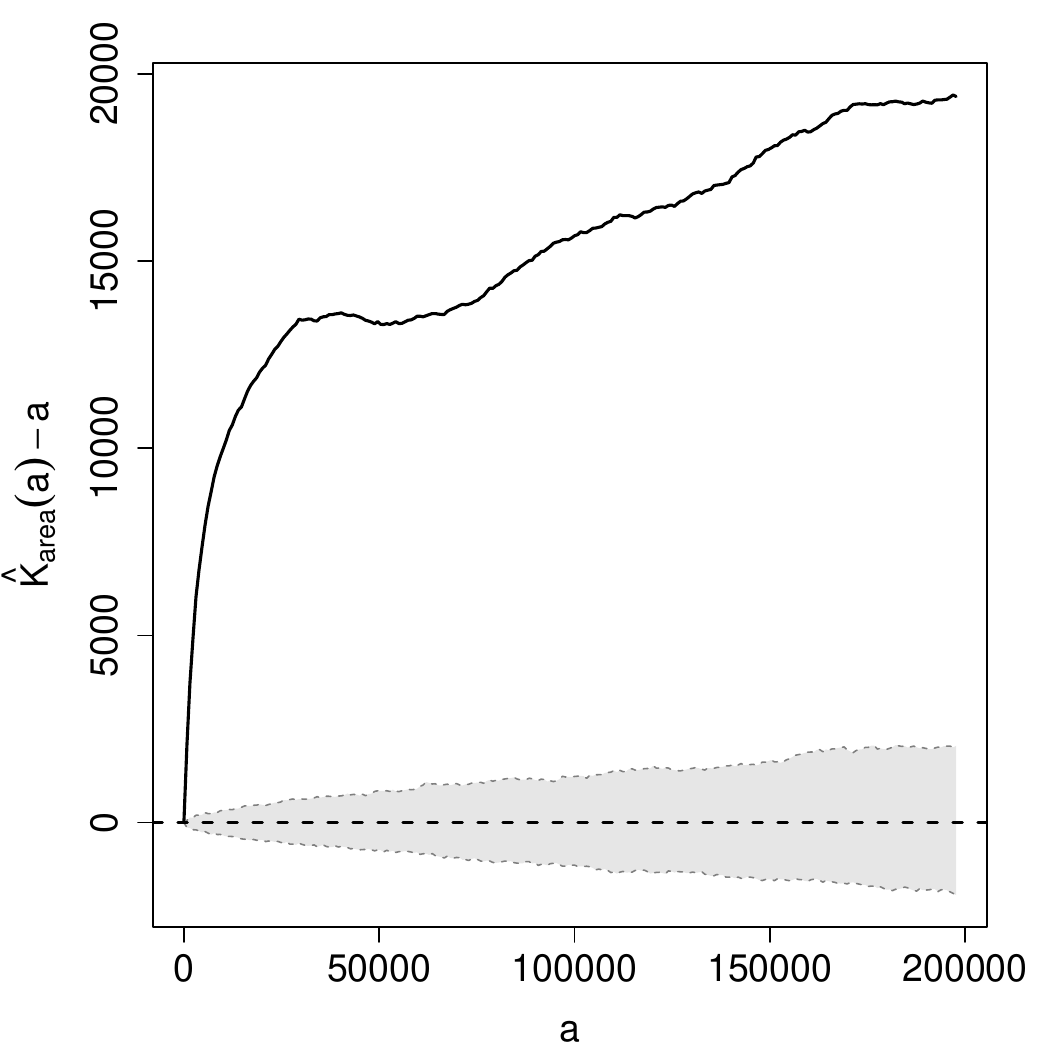}
\end{tabular}
    \caption{Left: view of the LPS dataset. The red dots are the
      LPS molecules, while the green curves on both sides of the pill
      delimit the observed region and the unobserved region due to
      the sampling procedure. Right: global envelope test.  Grey
      band: 95\% global envelope under the fitted homogeneous
      Poisson model. Black curve: centered surface area $\hat
      K_{\text{area}}(a)-a$-function. Dashed line: reference value zero under null hypothesis.}\label{fig:ecoli}
\end{figure}

\section{Discussion}

It is common practice in spatial statistics to treat locations
  of objects on a three-dimensional surface as planar point pattern,
  retaining only the $x$ and $y$ coordinates of the objects. We showed
  that this can lead to misleading inference regarding clustering
  properties due to spurious clustering when projecting points residing on a
  complex surface. It is thus crucial to define alternatives to the
  usual planar $K$-function that take into account the surface domain of the points.

We defined surface aware $K$-functions as empirical quantities that have
well-defined expectations under a Poisson process. They moreover from
a practical point of view make intuitive sense as summary statistics
for measuring deviations (clustering or regularity) from Poisson. It
would nevertheless be interesting to identify conditions that would make our surface
$K$-functions meaningful as estimates of population
characteristics for more general point processes.  
 For our tangent plane $K$-function with the logarithm mapping, it suffices
that the pair correlation function on the surface only depends on
geodesic distance between points. It is an open question how a similar
condition should be formulated for the surface area $K$-function. Also
it remains to construct specific models satisfying such
conditions.

Our proposed surface area $K$-function is a natural extension of the
classical $K$-function on Euclidean spaces. Interestingly, it
highlights the fact that area is a more natural argument for a $K$-function than radius,
since the expected number of points in a neighbourhood precisely
depends on the area of the neighbourhood. The surface area
$K$-function appealingly becomes the identity mapping under a Poisson
process.

\section*{Acknowledgments}

We gratefully acknowledge access to the first census data from the
Sinharaja research plot. The 25-ha Long-Term Ecological Research Project at Sinharaja World Heritage Site is a collaborative project of the Uva Wellassa University, the University of Peradeniya, the Forest Global Earth Observatory of the Smithsonian Tropical Research Institute, with supplementary funding received from the John D.\ and Catherine T.\ Macarthur Foundation, the National Institute for Environmental Science, Japan, and the Helmholtz Centre for Environmental Research-UFZ, Germany, for past censuses. The PIs gratefully acknowledge the Forest Department, Uva Wellassa University, and the Post-Graduate Institute of Science at the University of Peradeniya, Sri Lanka for supporting this project, and the local field and lab staff who tirelessly contributed in the repeated censuses of this plot. 

Francisco Cuevas-Pacheco was funded by the National Agency for Research and Development of Chile, under the grant ANID FONDECYT INICIACION 11240330.

\appendix

\section{Estimation of the intensity function of a surface point process}\label{sec:estimationintensity}

Suppose $X$ is observed within $S_W=\varphi(W) \subseteq S$ and that a parametric model $\ld_\ta$ is specified for the intensity
function. We can then estimate  $\ta$ using the unbiased estimating function
\[ \sum_{s \in X \cap S_W} \frac{\ld'(s;\ta)}{\ld(s;\ta)}- \int_{S_W}\ld'(t;\ta) \nu(\dd t). \]
Suppose we generate a dummy point process $D$ on $S_W$ of intensity
$\rho(\cdot)$ \citep{diaconis2013sampling}. Then the logistic regression approximation \citep{baddeley2014logistic} of the previous estimating function is
\begin{equation}\label{eq:logistic}
\sum_{s \in X \cap S_W} \frac{\ld'(s;\ta)}{\ld(s;\ta)}- \sum_{t \in (X
  \cup D) \cap S_W} \frac{\ld'(t;\ta)}{\ld(t;\ta)+\rho(t)}
\end{equation}
For the implementation of this we would use \texttt{glm} with offset
$-\log \rho(t)$. 

We could also consider the projected point processes $X_\pl$ and
$D_\pl=\varphi^{-1}(D)$ for which the logistic regression estimating function becomes
\[ \sum_{u \in X_\pl \cap W} \frac{\ld'_\pl(u;\ta)}{\ld_\pl(u;\ta)}- \sum_{u \in (X_\pl
  \cup D_\pl) \cap W} \frac{\ld_\pl'(u;\ta)}{\ld_\pl(u;\ta)+\rho_\pl(u)}
.\]
This is, however, exactly the same as \eqref{eq:logistic} after cancellation of $\|n(u)\|$.

\section{The Eikonal Equation}\label{sec:eikonal}

The \emph{Eikonal equation} is a nonlinear first-order partial differential equation
that describes the propagation of wavefronts through a medium
\citep{Sethian1999, Kimmel1998, Osher2003}. Let $s_{0} = \varphi(u_{0}) \in S$ be a
prescribed source point with $u_{0} \in D$, and let $\zeta : D \rightarrow \mathbb{R}$
be the function such that $\zeta(u)$ gives the geodesic distance $d_{S}(\varphi(u), s_{0})$
from $s_{0}$ to $\varphi(u)$. Then $\zeta$ satisfies

\begin{equation}
  \left\{ \begin{array}{l}
    \| \nabla \zeta(u) \|_{G}^{2} = 1, \\[4pt]
    \zeta(u_{0}) = 0,
  \end{array}\right.
  \label{eq:eikonal}
\end{equation}
where the squared gradient norm is taken with respect to the metric matrix $G(u)$
introduced in Section~\ref{sec:surfgeod}:
\begin{equation}
  \|\nabla \zeta(u)\|_{G}^{2}
  \;=\;
  \nabla \zeta(u)^{\top}\, G^{-1}(u)\, \nabla \zeta(u),
  \label{eq:grad-norm}
\end{equation}
and $\nabla \zeta(u) = (\partial_{u_{1}} \zeta(u),\, \partial_{u_{2}} \zeta(u))^{\top}$
is the gradient of $\zeta$ in the parameter coordinates.

Recall from Section~\ref{sec:surfgeod} that the $(i,j)$th entry of $G(u)$ is
\[
  G_{ij}(u) = \langle \partial_{u_{i}} \varphi(u),\, \partial_{u_{j}} \varphi(u) \rangle
  \;=\; \sum_{k=1}^{3} \partial_{u_{i}}\varphi_{k}(u)\, \partial_{u_{j}}\varphi_{k}(u),
  \qquad i,j = 1,2.
\]
The matrix $G(u)$ encodes all local geometric information of the surface, including
arc lengths, element areas, and angles. Its inverse $G^{-1}(u)$, whose $(i,j)$th entry
is denoted $G^{ij}(u)$, is the contravariant form of the metric at $u$. The solution
$\zeta(u)$ to \eqref{eq:eikonal}-\eqref{eq:grad-norm} recovers precisely the geodesic
distance $d_{S}(\varphi(u), s_{0})$, i.e.\ the length of the shortest path on $S$
connecting $\varphi(u)$ to $s_{0}$.

\subsection{The Fast Marching Method}

The Fast Marching Method \citep[introduced in][]{Sethian1996} exploits the fact that
the solution $\zeta$ of \eqref{eq:eikonal} grows monotonically outward from the source
$s_{0}$. The surface $S$ is approximated by $S_G =\varphi(D_G)$
where $D_G$ is a discrete grid covering $D$. In the developed
  code, $D_{G}$ is a regular grid and pairs
of grid nodes adjacent in the horizontal or vertical directions are
considered neighbours. The geodesic distance from $s_0=\varphi(u_0), u_0 \in D_G$, to every grid
node $t =\varphi(u)$, $u \in D_G$ is evaluated. Once a geodesic
distance is obtained for a node, no future update can decrease it, exactly
like Dijkstra's algorithm on a weighted graph but with a continuous PDE update rule
instead of edge weights. The algorithm is summarized in
Algorithm~\ref{alg:fmm}. Within the algorithm the following update rule is used.

\subsubsection*{The local update rule (Godunov upwind scheme)}

Consider a grid node corresponding to the parameter value $u = (u_1, u_2) \in D_G$, with
grid spacing $h_1$ and $h_2$ in the $u_1$- and $u_2$-directions respectively. Let
$\zeta^-_1 = \min(\zeta(u_{1}-h_1, u_2),\,\zeta(u_{1}+h_1, u_2))$ and
$\zeta^-_2 = \min(\zeta(u_{1}, u_2-h_2),\,\zeta(u_{1}, u_2+h_2))$ be the smallest
accepted neighbour values in each coordinate direction, and let $\sigma_1,\sigma_2 \in
\{-1,+1\}$ record whether $\zeta_1^-$ and $\zeta_2^-$ came from the
lower or upper neighbour respectively ($\sigma_i=+1$ if
  $\zeta_i^- = \zeta(u_i-h_i,u_{i})$, and $\sigma_i=-1$ if $\zeta_i^-
  = \zeta(u_i+h_i,u_{i})$). Then, the candidate value $\zeta=\zeta(u)$ solves
the quadratic equation
\citep{Kimmel1998}
\begin{equation}
  \frac{G^{11}(u)}{h_1^2}(\zeta - \zeta^-_1)^2
  + \frac{2\,\sigma_1\sigma_2\, G^{12}(u)}{h_1 h_2}(\zeta - \zeta^-_1)(\zeta - \zeta^-_2)
  + \frac{G^{22}(u)}{h_2^2}(\zeta - \zeta^-_2)^2 = 1,
  \label{eq:update-cross}
\end{equation}
taking the larger root. The signs $\sigma_1$ and $\sigma_2$ are needed because the one-sided difference approximating $\partial_{u_i}\zeta$ changes sign with the upwind
direction, and must be recomputed at every node since $\zeta_1^-,\zeta_2^-$ may each
come from either side.
The solution $\zeta$ is accepted only if $\zeta \geq \zeta^-_1$ and $\zeta \geq
\zeta^-_2$ (upwind condition). If this fails, or if only one of
$\zeta_1^-,\zeta_2^-$ is available (e.g.\ near a domain boundary), the update falls
back to a one-dimensional step along a single coordinate direction:
\begin{equation}
  \zeta = \min\!\left(\zeta_1^- + \frac{h_1}{\sqrt{G^{11}(u)}},\;\;
                        \zeta_2^- + \frac{h_2}{\sqrt{G^{22}(u)}}\right),
  \label{eq:update-1d}
\end{equation}
where a term is omitted from the minimum if the corresponding neighbour value is
unavailable. If the local metric is (near-)degenerate, \eqref{eq:update-cross} and
\eqref{eq:update-1d} are instead evaluated using unadjusted first-order neighbour
values.

\begin{algorithm}
\caption{Fast Marching Method}
\label{alg:fmm}
\DontPrintSemicolon
\KwIn{Grid $D_G$ with metric $G(u)$, source node $u_{0} \in D_G$}
\KwOut{Geodesic distance field $\zeta(u)$ for all nodes $u \in D_G$}
\BlankLine
\tcp{Initialisation}
Set $\zeta(u_{0}) \leftarrow 0$ and mark $u_{0}$ as \textsc{Accepted}\;
\For{each neighbour $u$ of $u_{0}$}{
  Compute candidate $\zeta(u)$ via \eqref{eq:update-cross}\;
  Mark $u$ as \textsc{Tentative} and insert into min-heap $\mathcal{H}$ with key $\zeta(u)$\;
}
Mark all remaining nodes as \textsc{Far} with $\zeta \leftarrow +\infty$\;
\BlankLine
\tcp{Main loop}
\While{$\mathcal{H} \neq \emptyset$}{
  $u^* \leftarrow$ \textsc{ExtractMin}$(\mathcal{H})$\;
  Mark $u^*$ as \textsc{Accepted}\;
  \For{each neighbour $u$ of $u^*$ not yet \textsc{Accepted}}{
    Compute candidate $\zeta_{\mathrm{new}}(u)$ via 
    \eqref{eq:update-cross} using all \textsc{Accepted} neighbours\;
    \uIf{$u$ is \textsc{Far}}{
      $\zeta(u) \leftarrow \zeta_{\mathrm{new}}(u)$\;
      Mark $u$ as \textsc{Tentative} and insert $u$ into $\mathcal{H}$\;
    }
    \ElseIf{$u$ is \textsc{Tentative} \textbf{and} $\zeta_{\mathrm{new}}(u) < \zeta(u)$}{
      $\zeta(u) \leftarrow \zeta_{\mathrm{new}}(u)$ and \textsc{DecreaseKey} in $\mathcal{H}$\;
    }
  }
}
\end{algorithm}

\newpage
\section*{Supplementary Material: Tangent Plane $K$-function}

\section{Introduction}\label{supp-sec:intro}

The tangent plane $K$-function is constructed by locally representing the surface in a Euclidean tangent space at each point of interest, and then defining the estimator in terms of Euclidean balls in that tangent space. In doing so, it provides a means of translating local geometric structure on the surface into a form that can be analysed using Euclidean tools.

To recap the exposition in the main text: the surface $S\subset
\mathbb{R}^3$ is defined as $S = \{s\in\mathbb{R}^3 :
s=\varphi(u)|u\in D\}$, where $D\subseteq\mathbb{R}^2$ is an open
subset, $\varphi:\mathbb{R}^2\mapsto\mathbb{R}^3$ is an injective
vector valued function $\varphi(u) =
(\varphi_1(u),\varphi_2(u),\varphi_3(u))$, and $u=(u_1,u_2)$ is of
dimension two. Let $T_s$ denote the tangent plane at $s\in S$, and
assume that $h_{ST}$ is a mapping (e.g.\ projection) from $S$ to
$T_s$. Suppose we observe $X$ inside $S_W = \varphi(W)$ for
  some $W\subseteq D$. We restrict the domain of $h_{ST}$ to a
connected neighbourhood $N_s \subset S_W$ of $s$ so that $h_{ST}$
becomes injective as a mapping from $N_s$ to $T_s$. We choose $N_s$ to
be the maximal such neighbourhood. The smallest geodesic distance from
the boundary of $N_s$ to $s$ is called the \emph{injectivity radius} at $s$, denoted $r^\ast_s$, and $r^\ast_S\equiv \inf_{s\in S}\ r^\ast_s$ is called the injectivity radius of $S$. 

We then let $X_{T_s} \subset A_{T_s}= h_{ST}(N_s)$ be the point process of points $h_{ST}(t)$, $t \in X \cap N_s$, and let $b_T(s,r) \subset T_s$ be the disc with
centre $s$ and radius $r$ contained in $T_s$.  Moreover, denote by $\ld_{T_s}$ the intensity function of $X_{T_s}$, whose derivation will later be demonstrated.
 Then, we define a tangent plane $K$-function as
\begin{equation}\label{supp-eq:tangentK} \hat K_{\text{tangent}}(r)=\frac{1}{\nu(S_W)} \sum_{s \in X \cap
		S_W} \sum_{v \in X_{T_s}\setminus h_{ST}(s)} \frac{1[ v \in b_T(s,r)]}{\ld(s) \ld_{T_s}(v)}e_T(s,r). \end{equation}
In this definition, $e_T(s,r)= \pi r^2/\nu( b_T(s,r) \cap A_{T_s})$
is  a correction factor which plays a role when $r$ is such that $b_T(s,r)\cap A_{T_s}\neq b_T(s,r)$.

Consider the case of a Poisson process and let \begin{align*} \hat K_s(r) & =  e_T(s,r) \sum_{t \in   X_{T_s} \setminus h_{ST}(s)} \frac{1[t \in b_T(s,r)]}{\ld_{T_s}(t)} .  \end{align*}
Then
\begin{align*} &\EE \hat K_s(r) = e_T(s,r) \int_{A_{T_s}}1[t \in  b_T(s,r)] \nu(\dd t) = \pi r^2 \end{align*}    which does not depend on $s$. We can rewrite our tangent plane $K$-function as
\[  \hat K_{\text{tangent}}(r)=\frac{1}{\nu(S_W)} \sum_{s \in X \cap
	S_W} \frac{1}{\ld(s)} \hat K_s(r) .\]
If $X$ is a Poisson process, it then follows by the  Slivnyak-Mecke theorem that
\[ \EE  \hat K_{\text{tangent}}(r) = \int_{S_W} \frac{1}{\nu(S_W)}
\frac{1}{\ld(s)}
\EE \hat K_S(r)  \ld(s)  \nu(\dd s) = \pi r^2.\]

We consider two choices for the mapping $h_{ST}$. The first is the
orthogonal projection onto $T_s$, which is computationally simple and
can be implemented directly for surfaces embedded in
$\mathbb{R}^3$. However, this mapping is not intrinsically adapted to
the geometry of $S$: Euclidean balls in $T_s$ do not, in general,
correspond to geodesic balls on the surface, and the resulting tangent
plane $K$-function therefore loses a direct geometric interpretation. The second choice is the logarithm map arising from a Riemannian formulation of $S$. This mapping is canonically associated with the intrinsic geometry of the surface and, locally, preserves geodesic distance from the base point, so that Euclidean balls of radius $r$ in $T_s$ correspond exactly to geodesic balls of radius $r$ on the surface. Furthermore, this framework allows us to move beyond the empirical formulation of the $K$-function in (\ref{supp-eq:tangentK}), to obtain a formal theoretical characterization. We treat these two cases in turn.

\section{Orthogonal mapping}\label{supp-sec:ortho}

We begin with the case where $h_{ST}$ is the orthogonal projection
onto the tangent plane. To derive the corresponding tangent space intensity, we first require some preliminaries on transformations between surfaces.

\subsection{Transformation between surfaces}\label{supp-sec:J_T}

Consider a general smooth surface $S$ in $\R^3$ parametrized by some injective function $\varphi:\R^2 \rightarrow S$.
By (8.8.6) in \cite{hoffman2017probability}, for $f: S \rightarrow [0,\infty[$,
\begin{equation}\label{supp-eq:r2S} \int_{\R^2} f(\varphi(u)) J(u) \dd u= \int_{S} f(s)  \nu(\dd s),  \end{equation}
where $\nu$ is surface measure (Hausdorff measure) on $S$ and
\[J(u)=\sqrt{\|\varphi_{x}(u)\|^2 \|\varphi_{y}(u)\|^2 - (\varphi_{x}(u)^\T \varphi_{y}(u))^2},\]
where $\varphi_x(u)=\partial \varphi/\partial x$ and $\varphi_y(u)=\partial \varphi/\partial x$. If $\varphi(u)=(u,z(u))$ for some function
$z:\R^2  \rightarrow \R$, then $\varphi_x(u)=(1,0,z_x(u))^\T$ and
$\varphi_y(u)=(0,1,z_y(u))^\T$ in which case
$J(u)=\sqrt{1+z_x^2(u)+z_y^2(u)}$.  

\subsection{Intensity of point process on tangent plane}\label{supp-sec:tangentplaneintensity}

Let $T_s$ denote the tangent plane at the point $s \in S$ and
denote the resulting point process projected onto $T_s$ by $X_{T_s}$ with intensity $\ld_{T_s}$ (with respect to surface measure on
$T_s$). Then we have for $A \subseteq T_s$, 
\begin{align*} \EE n(X_{T_s} \cap A) = \EE \sum_{s \in X} 1[ h_{ST}(s) \in A] =& \int_{S} 1[h_{ST}(s) \in A] \ld(s) \nu( \dd s) \\ =
	&\int_{T_s} 1[ v \in A] \ld_{T_s}(v) \nu (\dd v) .\end{align*}
Note that we can parametrize $T_s$ by the composite mapping
$h_T(u)=h_{ST}(\varphi(u))$. Then, applying \eqref{supp-eq:r2S} twice,
we obtain
\[ \int_{\R^2} 1[h_{ST}(\varphi(u)) \in A] \ld(\varphi(u)) J(u) \dd u =
\int_{\R^2} 1[ h_{T}(u) \in A] \ld_{T_s}(h_{T}(u) )J_{T}(u)
\dd u 
\]
where $J_{T}(u)$ is the Jacobian for the mapping $h_{T}(u)$. Thus we
can conclude
\[ \ld_{T_s}(h_{T}(u)) = \ld(\varphi(u)) \frac{J(u)}{J_T(u)} \]
or 
\[ \ld_{T_s}(v) = \ld(h_{ST}^{-1}(v)) \frac{J(h_T^{-1}(v))}{J_T(h_T^{-1}(v))}.\]

\subsection{Calculation of $J_T$ in case of orthogonal projection and topograpic surface}\label{supp-sec:jacobian}

We have 
\[ J_T(u)= \sqrt{\|h_{T,x}(u)\|^2 \|h_{T,y}(u)\|^2 - (h_{T,x}(u)^\T
	h_{T,y}(u))^2} \]
where 
\[ h_{T,x}(u)= h'_{ST}(\varphi(u)) (1,0,z_x(u))^\T \quad h_{T,y}(u)=
h'_{ST}(\varphi(u)) (0,1,z_y(u))^\T \]
and $h'_{ST}(t)$ is the matrix of partial derivatives $\partial h_{ST,i}/
\partial t_j$. Suppose $h_{ST}(t)=p_T(t)$, the orthogonal projection of $t$ on $T_s$. The orthogonal projection can be expressed as
\[ p_T(t)= t - \frac{(t-s)^\T n_s}{ \|n_s\|^2} n_s =
\left (I- \frac{n_s n_s^\T}{ \|n_s\|^2} \right)t -
\frac{n_s n_s^\T}{ \|n_s\|^2} s \]
where $n_s= (-z_x(\varphi^{-1}(s)),-z_y(\varphi^{-1}(s)),1)$ is the surface normal vector at $s$.
Hence, 
\[ h_{ST}'(t) = I- \frac{n_s n_s^\T}{ \|n_s\|^2} .\]
Computing the norm of the cross product of the resulting $h_{T,x}$ and
$h_{T,y}$ results in
\[ J_T(u)= \frac{|n_s^\T n(u)|}{ \|n_s\|}, \]
where $n_{\varphi(u)}=n(u)= (-z_x(u),-z_y(u),1)^\T$ is the surface normal vector at
$\varphi(u)$. Thus
\[ \frac{J_T(u)}{J(u)}= \frac{|n_s^\T n_{\varphi(u)}|}{ \|n_s \|
	\|n_{\varphi(u)})\|} = |\cos(\theta)|, \]
where $\theta$ is the angle between $n_s$ and $n_{\varphi(u)}$.

\section{Riemannian formulation and the Exp/Log mapping}\label{supp-sec:riemannian}

We now turn to the Riemannian formulation. The smooth surface $S \subset \mathbb{R}^3$ may be viewed as a 2-dimensional Riemannian manifold equipped with the metric induced by its embedding, and the associated geodesic distance agrees with that used in the main text. In what follows, we write $(\cM,\rg)$ for a general Riemannian manifold, with the understanding that in the setting of the main manuscript $\cM=S$ endowed with its induced metric. This provides the setting for the logarithm and exponential maps used below. We next summarize the required definitions and notation.

\subsection{Riemannian manifolds}

Let $(\cM,\rg)$ be a smooth orientated $d$-dimensional Riemannian manifold with metric $\rg$ without boundary and for $s\in\cM$ denote $T_{s}$ to be the tangent space at $s$. Let $(U,\psi)$ be a local chart, where $U\subset\cM$ and $\psi:U\rightarrow \mathbb{R}^d$ is a homeomorphism such that $\psi(s) = (x^1(s),...,x^d(s))$ for $s\in\cM$. Metric $\rg$ can be written locally at $s\in\cM$ as
\begin{equation*}
	\rg_s = \sum_{i,j=1}^{d} \rg_{ij}(s)\,dx^i \otimes dx^j,
\end{equation*}
where $\otimes$ is the tensor product and $\rg_{ij}:U\mapsto\mathbb{R}$ is a collection of $d^2$ smooth functions such that $\rg_{ij}(s)=\langle \partial/\partial x^i|_{s}, \partial/\partial x^j|_{s} \rangle_s$, where $\langle \cdot,\cdot \rangle_s$ represents the inner-product on the tangent space $T_{s}$ assigned by $\rg$ \citep{Lee2018}. 

We denote the Riemannian volume form, or just volume form, as $d\text{vol}$, where
\begin{equation*}
	d\text{vol}= \sqrt{\det(\rg_{ij})} dx^1\wedge\cdots\wedge dx^d= \sqrt{\det(\rg_{ij})} dx,
\end{equation*}
where $\wedge$ is the wedge product and $dx=dx^1\wedge\cdots\wedge dx^d$ \citep{Lee2018}. We define the volume of compact subset $B\subset\cM$ as $\text{Vol}(B)=\int_B\, d\text{vol}(x)$. 

The geodesic from point $s$ to a point $t\in\cM$ can be informally considered as the curve connecting these points such that the \emph{distance} is minimal. Formally, consider the set of continuously differentiable curves $\gamma:[a,b] \mapsto \cM$ with $\gamma(a)=s$ and $\gamma(b)=t$. Then the geodesic $\gamma'$ is the curve that minimises
\begin{equation*}
	L_\gamma = \int_a^b \sqrt{\rg_{\gamma(\tau)}\left(\frac{d\gamma(\tau)}{d\tau},\frac{d\gamma(\tau)}{d\tau}\right)} d\tau.
\end{equation*}
The geodesic distance between points $s$ and $t$ is taken as the minimum value of $L$ over all possible curves $\gamma$ and denoted $d_{\cM,\rg}:\cM\times\cM\mapsto \mathbb{R}$. Note that $(\cM,d_{\cM,\rg})$ defines a metric space.

Let $v\in T_{s}$ where $s\in\cM$. Then there exists a unique geodesic $\gamma_{s,v}$ such that $\gamma_{s,v}(0)=s$ and $d\gamma_{s,v}(\tau)/d\tau|_{\tau=0}=v$. The exponential mapping $\exp: T_{s}\rightarrow\cM$ is defined as $\exp_{s}(v) = \gamma_{s,v}(1).$ Furthermore, the logarithm map is defined as the inverse of the exponential map, i.e. $\log_{s}(t)= \exp^{-1}_{s}(t)$ where $t\in\cM$.

Following \cite{Willmore1996}, we can define the volume density function as,
\begin{equation*}
	\theta_{s}(t)=\left.\frac{\exp^*_{s}d\text{vol}(x)}{\mu_\rg(dx)}\right|_{x=\log_{s}(t)}
\end{equation*}
for $s,t\in\cM$, where $\exp^*$ is the pullback metric under the exponential map and $\mu_\rg$ is the Lebesgue measure on $T_s$ induced by $\rg$. By imposing that the injectivity radius $r_{\cM,\rg}^*$ of $(\cM,\rg)$ is strictly positive, the volume density function is at least well defined for all $t$ lying in a neighbourhood of $s$. Using Jacobi fields, \cite{Willmore1996} shows that the volume density function can be defined for all $t\in\cM$ instead of in a neighbourhood of $s\in\cM$. Moreover, if we consider the specific chart $(U,\log_{s})$, such that $U$ is a compact neighbourhood of $s\in\cM$, then the volume density function can be written as
\begin{equation*}
	\theta_{s}(t) = \sqrt{\det[\rg_{ij}(\log_{s}(t))]}.
\end{equation*}
Let $(U,\psi)$ be a chart of $\cM$ and $f$ a
continuous function over a compactly supported subset of $U$. Then,
\begin{equation*}
	\int_U f(s)\,d\text{vol}(s) = \int_{\psi(U)} (\psi^{-1})^*f(s)\sqrt{\det(\rg_{ij})} ds,
\end{equation*}
where $(\psi^{-1})^*f$ is the pullback of $f$ by $\psi^{-1}$ and $\rg_{ij}$ is the Riemannian metric expressed locally in the coordinates given by $(U,\psi)$.

\subsection{Point processes on Riemannian manifolds}

We denote $X$ to be a point process over the metric space
$(\cM,d_{\cM,\rg})$ with its accompanying random counting measure
$N_X(B)$ which counts the number of points of $X$ in $B\subset
\cM$. We suppose that $X$ is locally finite and simple, i.e. that
$N_X(B)<\infty$ for any compact set $B$ and points of $X$ are not
coincidental respectively. The intensity measure of $X$ is
$\mu(B)=\mathbb{E}[N_X(B)]$ and if $\mu$ is absolutely continuous with
respect to the volume form then it can be written as
\begin{equation*}
	\mu(B)=\int_{B}\lambda(s)d\text{vol}(s),
\end{equation*}
where $\lambda$ is known as the intensity of $X$, which we assume to exist. We also define the pair correlation function (pcf) of $X$ as,
\begin{equation*}
	g(s,t)=\frac{\lambda^{(2)}(s,t)}{\lambda(s)\lambda(t)},
\end{equation*}
where $\lambda^{(2)}$ is the second order intensity function defined as the density of $\mathbb{E}\sum_{s,t\in X}^{\neq}1[s\in B_1,t\in B_2]$ with respect to the product of the volume forms $d\text{vol}(s)\times d\text{vol}(t)$, i.e. $\lambda^{(2)}:\cM\times\cM\mapsto\mathbb{R}^+$ is the function such that,
\begin{equation*}
	\mathbb{E}\mathop{\sum\nolimits\sp{\ne}}_{s,t\in X}1[s\in B_1,t\in B_2]=\int_{B_1}\int_{B_2}\lambda^{(2)}(s,t)d\text{vol}(s)d\text{vol}(t),
\end{equation*}
where $B_1,B_2\subset \cM$ such that $\text{Vol}(B_1)>0$ and $\text{Vol}(B_2)>0$. The Campbell-Mecke Theorem for point processes on Riemannian manifolds is
\begin{equation*}
	\mathbb{E}\sum_{s\in X} f(s,X\backslash s) = \int_{\cM} \mathbb{E}[f(s,X^{!}_{s})]\lambda(s) d\text{vol}(s),
\end{equation*}
where $f$ is a non-negative measurable function and $X^{!}_{s}$ is the reduced Palm process of $X$ \citep{Moller2003}. We say that a point process $X$ is Poisson with intensity $\lambda:\cM\mapsto\mathbb{R}$ if for any $B\subset\cM$ such that $\text{Vol}(B)>0$, the number of points of $X$ in $B$ is Poisson with mean $\int_B\lambda(s)d\text{vol}(d s)$ and given the number of points, they are independently and identically distributed across $B$ with density $\lambda(s)/\mu(B)$.

\subsection{Exp/log mapping}

When constructing the tangent $K$-function in the main manuscript we require a tractable mapping from the original space to a tangent space at any given point. While the orthogonal mapping is computationally efficient to work with, as discussed, Euclidean balls on the tangent space do not, in general, translate back to geodesic balls on the manifold. Instead, a natural choice for Riemannian manifolds is to set the mapping $h_{\cM,T}:\cM\rightarrow T_s$ to be the logarithm mapping, denoted $\log_s$, with its inverse, the exponential mapping, denoted $\exp_s$. In particular, noting that $\log_s(s)=0$ and that $d_{\cM,\rg}(s,t)=\|\log_s(t)\|_{s}$ \citep{doCarmo1992}, the map $\log_s$ sends the geodesic ball $b_{\cM}(s,r)$ on the manifold to the Euclidean ball $b_{T_s}(0,r)$ on the tangent space, or equivalently to the Euclidean ball $b_{\mathbb{R}^d}(0,r)$ in $\mathbb{R}^d$. The latter has volume $\pi^{d/2}r^d/\Gamma(1+d/2)$, irrespective of the Riemannian volume of $b_{\cM}(s,r)$ on $\cM$. 

For point processes on $\mathbb{R}^d$, there is the requirement that
the process be \emph{second order intensity reweighted stationary} for
the $K$-function to be well defined, in the global sense
\citep{Baddeley2000}. This allows for inhomogeneity by requiring that
the pcf of the process is translationally invariant. Analogously, we require a corresponding assumption on the first- and second-order structure of our process in order to provide a well-defined $K$-function for point processes on Riemannian manifolds. 

To define a valid $K$-function on linear networks, \citet{rakshit2017} introduce $\delta$\emph{-correlated} processes. Similarly, we define the $d_{\cM,\rg}$\emph{-correlated} class of point processes on Riemannian manifold $(\cM,\rg)$. 

\begin{definition}\label{supp-def:d-correlated}
	Let $X$ be a point process on Riemannian manifold $(\cM,\rg)$. We say that $X$ is $d_{\cM,\rg}$-correlated if,
	\begin{equation*}
		g(s,t)=g_0(d_{\cM,\rg}(s,t)),
	\end{equation*}
	where $d_{\cM,\rg}$ is the distance metric induced by $\rg$ and $g_0:\mathbb{R}^+\rightarrow\mathbb{R}^+$.
\end{definition}

It is easily shown that a Poisson process, with either
constant or varying intensity function, has a pcf equal to 1 and so a
Poisson process is therefore $d_{\mathcal M,\rg}$-correlated for any
Riemannian metric $\rg$. Furthermore, suppose one is able to
define a Gaussian random field $\{Z(\cdot)\}$ on $\cM$ such
that $\cov\{Z(s),Z(t)\}$ depends only on
$d_{\cM,\rg}(s,t)$. Then a log Gaussian Cox process with
random intensity function $\exp(Z(\cdot))$ would be
$d_{\cM,\rg}$-correlated too.

\subsubsection*{$K$-function}

We define the $K$-function for $d_{\mathcal M,\rg}$-correlated point processes using similar arguments as that made in Definition 5 of \cite{rakshit2017}.

\begin{definition}\label{supp-def:Kg}
	Let $X$ be a $d_{\mathcal M,\rg}$-correlated point process on $(\cM,\rg)$ with intensity function $\lambda$ that is positive and locally integrable over $\cM$. The $K$-function with respect to $(\cM,\rg)$ is
	\begin{equation}\label{supp-eq:K:homo}
		K_{\cM,\rg}(r)
		=
		\frac{1}{\mathrm{Vol}(A)}
		\mathbb{E}\mathop{\sum\nolimits\sp{\ne}}_{s,t\in X}
		\frac{1[s\in A,\; t\in b_{\cM}(s,r)]}{\lambda(s)\lambda(t)\theta_s(t)},
	\end{equation}
	for any $A\subset \cM$ of positive volume and $0\leq r<r^*_{\cM,\rg}$.
\end{definition}

We note that a natural estimator for $K_{\cM,\rg}(r)$ is exactly $\hat K_{\text{tangent}}(r)$ as defined in (\ref{supp-eq:tangentK}), recognizing that $\lambda_{T_s}(t)= \lambda(t)\theta_s(t)$ and $e_{T,S}(r) = 1$ when $0\leq r<r^*_{\cM,\rg}$.

\begin{proposition}\label{supp-prop:K:g:relation}
	Let $X$ be a $d_{\mathcal M,\rg}$-correlated point process on $(\cM,\rg)$ with positive and locally integrable intensity function $\lambda$. Then $K_{\cM,\rg}$ and the pair correlation function $g_0$ are related by
	\begin{equation*}
		K_{\cM,\rg}(r)=\int_{b_{\mathbb{R}^d}(0,r)}g_0(\|v\|)\,dv,
	\end{equation*}
	where $0<r<r^*_{\cM,\rg}$, $dv \equiv \mu_\rg(dv)$ is the Lebesgue measure on $\mathbb{R}^d$, and $\|\cdot\|$ denotes the Euclidean norm on $\mathbb{R}^d$.
\end{proposition}

\begin{proof}
	By Campbell's theorem,
	\begin{align*}
		K_{\cM,\rg}(r)
		&=\frac{1}{\text{Vol}(A)}\mathbb{E}\mathop{\sum\nolimits\sp{\ne}}_{s,t\in X}\frac{1[s\in A, t\in b_{\cM}(s,r)]}{\ld(s)\ld(t)\theta_{s}(t)}\\
		&=\frac{1}{\mathrm{Vol}(A)}
		\int_A\int_{b_{\cM}(s,r)}
		\frac{\lambda^{(2)}(s,t)}{\lambda(s)\lambda(t)\theta_s(t)}
		\,d\mathrm{vol}(t)\,d\mathrm{vol}(s).
	\end{align*}
	Since $X$ is $d_{\mathcal M,\rg}$-correlated,
	\begin{align*}
		K_{\cM,\rg}(r)
		&=
		\frac{1}{\mathrm{Vol}(A)}
		\int_A\int_{b_{\cM}(s,r)}
		\frac{g_0(d_{\cM,\rg}(s,t))}{\theta_s(t)}
		\,d\mathrm{vol}(t)\,d\mathrm{vol}(s).
	\end{align*}
	Now fix $s\in A$. Since $r<r^*_{\cM,\rg}$, the exponential map
	$
	\exp_s : b_{T_s}(0,r)\to b_{\cM}(s,r)
	$
	is a diffeomorphism. Hence, using the change of variables $t=\exp_s(v)$, we have $d\mathrm{vol}(t)=\theta_s(t)\,dv,$ where $dv$ denotes Lebesgue measure on $T_s$ under an orthonormal identification $T_s\cong\mathbb{R}^d$. Moreover, $d_{\cM,\rg}(s,\exp_s(v))=\|v\|_{s}$, where $\|\cdot\|_s^2\equiv \langle\cdot,\cdot\rangle_s$.
	Therefore,
	\begin{align*}
		\int_{b_{\cM}(s,r)}
		\frac{g_0(d_{\cM,\rg}(s,t))}{\theta_s(t)}
		\,d\mathrm{vol}(t)
		&=
		\int_{b_{T_s}(0,r)}
		g_0(\|v\|_{s})\,dv \\
		&=
		\int_{b_{\mathbb{R}^d}(0,r)}
		g_0(\|v\|)\,dv.
	\end{align*}
	The final expression is independent of $s$, so
	\begin{align*}
		K_{\cM,\rg}(r)
		&=
		\frac{1}{\mathrm{Vol}(A)}
		\int_A
		\left(
		\int_{b_{\mathbb{R}^d}(0,r)}g_0(\|v\|)\,dv
		\right)
		d\mathrm{vol}(s) \\
		&=
		\int_{b_{\mathbb{R}^d}(0,r)}g_0(\|v\|)\,dv.
	\end{align*}
	This proves the result.
\end{proof}

Definition \ref{supp-def:Kg} and Proposition \ref{supp-prop:K:g:relation} naturally extends Ripley's $K$-function for point processes on a Euclidean surface to general Riemannian manifolds. Additionally, by the Campbell Theorem we can also define the $K$-function with respect to $\rg$ in terms of the reduced Palm process $X^{!}_{s}$ as,
\begin{equation}
	\label{supp-eq:Kpalm}
	K_{\cM,\rg}(r) =\mathbb{E}\mathop{\sum\nolimits\sp{\ne}}_{t\in X^!_{s}}\frac{1[t\in b_{\cM}(s,r)]}{\theta_{s}(t)\lambda(t)},
\end{equation}
where the right hand side holds for any $s\in\cM$.

\begin{corollary}\label{supp-corr:K:homo:PPP}
	Let $X$ be a Poisson process on $d$-dimensional manifold $(\cM,\rg)$ with positive and locally integrable intensity function $\lambda(\cdot)$. Then 
	\begin{equation*}
		K_{\cM,\rg}(r) = \omega_d r^d,
	\end{equation*}
	where $\omega_d=\pi^{d/2}/\Gamma(1+d/2)$, i.e $\omega_d$ is the volume of a $d$-dimensional Euclidean ball.
\end{corollary}

\begin{proof}
	This follows immediately since $g_0(d_{\cM,\rg}(s,t))=g(s,t)=1$ for any Poisson process.
\end{proof}

\begin{remark}
	Note that in the case of $X$ being a homogeneous $d_{\mathcal M,\rg}$-correlated process, the intensity function $\lambda(\cdot)$ in Definition \ref{supp-def:Kg}, Proposition \ref{supp-prop:K:g:relation} and (\ref{supp-eq:Kpalm}) can be replaced with the constant intensity $\lambda>0$.
\end{remark}

\bibliographystyle{apalike}
\bibliography{surface}

\end{document}